\documentclass[10pt]{article}
\font\smallit=cmti10

\usepackage{amssymb,latexsym,amsmath,epsfig,amsthm} 
\usepackage{amsthm}
\usepackage{paithan}
\usepackage{xcolor}
\usepackage{mathtools}
\usepackage{tikz}
\usepackage{enumitem}
\usetikzlibrary{arrows.meta, positioning}
\usetikzlibrary{shapes, fit, backgrounds, shapes.geometric}
\definecolor{mypurple}{RGB}{218,112,214}
\definecolor{mygreen}{RGB}{0,128,0}
\definecolor{transitionpurple}{RGB}{170, 102, 254}
\usepackage{braket}
\usepackage{nicematrix}

\makeatletter

\renewcommand\section{\@startsection {section}{1}{\z@}
{-30pt \@plus -1ex \@minus -.2ex}
{2.3ex \@plus.2ex}
{\normalfont\normalsize\bfseries\boldmath}}

\renewcommand\subsection{\@startsection{subsection}{2}{\z@}
{-3.25ex\@plus -1ex \@minus -.2ex}
{1.5ex \@plus .2ex}
{\normalfont\normalsize\bfseries\boldmath}}

\renewcommand{\@seccntformat}[1]{\csname the#1\endcsname. }

\makeatother

\newtheorem{theorem}{Theorem}

\newtheorem{proposition}{Proposition}
\newtheorem{corollary}{Corollary}
\newtheorem{gameruleset}{Ruleset}
\theoremstyle{definition}
\newtheorem{definition}{Definition}

\begin{document}

\begin{center}
\uppercase
{\bf Transverse Wave: an impartial color-propagation game inspried by Social Influence and Quantum Nim}
\vskip 20pt
{\bf Kyle Burke}\\
{\smallit Department of Computer Science, Plymouth State University, Plymouth, New Hampshire, {\tt \url{https://turing.plymouth.edu/~kgb1013/}}}\\
\vskip 10pt
{\bf Matthew Ferland}\\
{\smallit Department of Computer Science, University of Southern California, Los Angeles, California, United States}\\
{\tt \url{}}\\
{\bf Shang-Hua Teng\footnote{Supported by the Simons Investigator Award for fundamental \& curiosity-driven research and NSF grant CCF1815254.}}\\
{\smallit Department of Computer Science, University of Southern California, Los Angeles, California, United States}\\
{\tt \url{https://viterbi-web.usc.edu/~shanghua/}}\\
\end{center}
\vskip 20pt
\vskip 30pt

\centerline{\bf Abstract}
\noindent
In this paper, we study a colorful, impartial combinatorial game played on a two-dimensional grid, {\sc Transverse Wave}. 
We are drawn to this game because of its apparent simplicity, 
contrasting intractability, and intrinsic connection to two other combinatorial games, one inspired by social influence 
 and another inspired by quantum 
superpositions. 
More precisely, we show that {\sc Transverse Wave}
  is at the intersection of social-influence-inspired
{\sc Friend Circle} and superposition-based {\sc Demi-Quantum Nim}.
{\sc Transverse Wave} is also connected with Schaefer's 
  logic game {\sc Avoid True}.
In addition to analyzing the mathematical structures and computational complexity of {\sc Transverse Wave}, 
we provide a web-based version of the game, playable at
{\tt \url{https://turing.plymouth.edu/~kgb1013/DB/combGames/transverseWave.html.}}
Furthermore, we formulate a basic network-influence inspired game,
called  {\sc Demographic Influence}, which simultaneously 
generalizes {\sc Node-Kyles} and {\sc Demi-Quantum Nim} (which in turn contains as special cases {\sc Nim}, {\sc Avoid True}, and {\sc Transverse Wave}.).
These connections illuminate the \textit{lattice order}, induced by special-case/generalization relationships over mathematical games,
fundamental to both the design and comparative analyses of combinatorial games.

\pagestyle{myheadings}
\thispagestyle{empty}
\baselineskip=12.875pt
\vskip 30pt

\section{The Game Dynamics of {\sc Transverse Wave}}



Elwyn Berlekamp, John Conway, and Richard Guy were known 
  not only for their deep mathematical discoveries 
  but also for their elegant minds. 
Their love of mathematics and their life-long efforts of making mathematics fun and approachable had led to their master piece, ``Winning Ways for your Mathematical Plays," \cite{WinningWays:2001},
a book that has inspired many.
The field that they pioneered---{\em combinatorial game theory}---reflects their personalities. 
Illustrious combinatorial games are usually: 
\begin{itemize}
\item {\bf\em Approachable}:  having a simple, easy to remember and easy to understand ruleset, and 
\item {\bf\em Elegant}: having attractive game boards, yet 
\item {\bf\em  Intriguing \& Challenging}:  having rich strategy structures and requiring  nontrivial efforts to play well (optimally).
\end{itemize}

\subsection{Combinatorial Games and Computational Complexity}

The last property has a characterization using computational complexity.
As the dimensions of games grow---illustrated by popular games from {\sc Go} \cite{LichtensteinSipser:1980} to {\sc Hex} \cite{DBLP:journals/jcss/Schaefer78,Reisch:1981}---the underlying computational problem for determining
the outcome of their game positions and 
for selecting winner moves becomes computational intractable.
Thus, elegant combinatorial games with simple, easy to understand \& remember rulesets yet intractable complexity are the gold standard for combinatorial game design \cite{Eppstein,BurkeFerlandTengQCGT}.

The computational complexity of determining the winnability of a ruleset is a common problem. For games where the winner can be determined algorithmically in polynomial-time, optimal players can be programmed that will run efficiently. In a match with one of these players, there is no need to play the game out to determine whether you can win; just run the algorithm and see what it tells you.

In order to make the competition \emph{interesting}, we want the winnability to be \emph{computationally intractable}, meaning there's no known efficient algorithm to always calculate a position's outcome class. One way to argue that this is the case is to show that the problem of finding the outcome class is hard for a common complexity class.  Many such combinatorial games are found to be \cclass{PSPACE}-hard, meaning that finding a polynomial-time algorithm automatically leads to a  polynomial-time solution to \emph{all} problems in \cclass{PSPACE}.

This argument does not entirely settle the debate about whether a ruleset is "interesting". Indeed, it could be the case that from common starting positions, there is a strategy for the winning player to avoid the computationally-hard positions.  Finding computational-hardness for general positions in a ruleset is only a mininum-requirement.  Improvements can be made by finding hard positions more likely to result after game play from the start.  The best proofs of hardness yield positions that are starting positions themselves\footnote{This requires some variance on these starting positions.  "Empty" or well-structured initial boards do not have a large enough descriptive size to be computationally hard in the expected measures.}.


\subsection{A Colorful Propagation Game over 2D Grids}

In this paper, we consider a simple, colorful, impartial combinatorial game over two-dimensional grids.\footnote{We tested the approachability of this game by explaining its ruleset to a bilingual eight-year-old second-grade student---in Chinese---and she turned to her historian mother and flawlessly explained the ruleset in English.}
We call this game {\sc Transverse Wave}.
We become interested in this game during our study of 
quantum combinatorial games \cite{BurkeFerlandTengQCGT}, particularly, in  our complexity-theoretical analysis of a family of games formulated 
   by superposition of the classical {\sc Nim}.
In addition to Schaefer's logic game, {\sc Avoid True},
{\sc Tranaverse Wave} is also fundamentally connected with several social-influence-inspired combinatorial games.  As we shall show, \ruleset{Transverse Wave}, is \cclass{PSPACE}-hard on some possible starting positions.

\begin{gameruleset}[{\sc Transverse Wave}]
For a pair of integer parameters $m, n > 0$, a game position of {\sc Transverse Wave} is an $m$ by $n$ grid $G$, in which the cells are colored either \texttt{green} or \texttt{purple}.\footnote{or any pair of easily-distinguishable colors.}

For the game instance starting at this position, 
 two players take turns selecting a column of this colorful grid. 
A column $j\in [n]$ is feasible for $G$ if it 
  contains at least one \texttt{green} cell.
The selection of $j$ transforms $G$ into another colorful 
  $m$ by $n$ grid $G\otimes [j]$ by 
  recoloring column $j$ and every row with a \texttt{purple} cell 
  at column $j$ with \texttt{purple}.\footnote{Think of purple paint cascading down column $j$ and inducing 
    a purple "transverse wave" whenever the propagation 
    goes through an already-purple cell.}
In the normal-play convention, the player without a feasible 
  move loses the game.
\end{gameruleset}

\begin{figure}[h!]
\begin{tikzpicture}[box/.style={rectangle,draw=black,thick, minimum size=1cm}]

  \foreach \x in {0,1,...,3}{
      \foreach \y in {0,1,...,4}
          \node[box, fill=mygreen] at (\x,\y){};
  }
  
  \node at (0,5) {0};
  \node at (1,5) {1};
  \node[circle] at (2,5) {2};
  \node at (3,5) {3};
  
  \node[box,fill=mypurple] at (0,4){}; 
  \node[box,fill=mypurple] at (1,3){}; 
  \node[box,fill=mypurple] at (2,3){}; 
  \node[box,fill=mypurple] at (1,2){}; 
  \node[box,fill=mypurple] at (1,1){}; 
  \node[box,fill=mypurple] at (2,1){}; 
  \node[box,fill=mypurple] at (3,1){}; 
  \node[box,fill=mypurple] at (0,0){};  
  \node[box,fill=mypurple] at (1,0){};  
  \node[box,fill=mypurple] at (3,0){};  
  
  \node at (1.5, -1) {a};

\end{tikzpicture}
\begin{tikzpicture}[box/.style={rectangle,draw=black,thick, minimum size=1cm}]

  \foreach \x in {0,1,...,3}{
      \foreach \y in {0,1,...,4}
          \node[box, fill=mygreen] at (\x,\y){};
  }
  \node at (0,5) {0};
  \node at (1,5) {1};
  \node[draw, circle] at (2,5) {\textbf{2}};
  \node at (3,5) {3};
  
  \node[box,fill=mypurple] at (0,4){}; 
  \node[box,fill=transitionpurple] at (2,4){};
  \node[box,fill=transitionpurple] at (0,3){};
  \node[box,fill=mypurple] at (1,3){}; 
  \node[box,fill=mypurple] at (2,3){};
  \node[box,fill=transitionpurple] at (3,3){};
  \node[box,fill=mypurple] at (1,2){};
  \node[box,fill=transitionpurple] at (2,2){};
  \node[box,fill=transitionpurple] at (0,1){};
  \node[box,fill=mypurple] at (1,1){}; 
  \node[box,fill=mypurple] at (2,1){}; 
  \node[box,fill=mypurple] at (3,1){}; 
  \node[box,fill=mypurple] at (0,0){};  
  \node[box,fill=mypurple] at (1,0){}; 
  \node[box,fill=transitionpurple] at (2,0){};
  \node[box,fill=mypurple] at (3,0){};  
  
  \node at (1.5, -1) {b};

\end{tikzpicture}
\begin{tikzpicture}[box/.style={rectangle,draw=black,thick, minimum size=1cm}]

  \foreach \x in {0,1,...,3}{
      \foreach \y in {0,1,...,4}
          \node[box, fill=mygreen] at (\x,\y){};
  }
  \node at (0,5) {0};
  \node at (1,5) {1};
  \node at (2,5) {2};
  \node at (3,5) {3};
  
  \node[box,fill=mypurple] at (0,4){}; 
  \node[box,fill=mypurple] at (2,4){};
  \node[box,fill=mypurple] at (0,3){};
  \node[box,fill=mypurple] at (1,3){}; 
  \node[box,fill=mypurple] at (2,3){};
  \node[box,fill=mypurple] at (3,3){};
  \node[box,fill=mypurple] at (1,2){};
  \node[box,fill=mypurple] at (2,2){};
  \node[box,fill=mypurple] at (0,1){};
  \node[box,fill=mypurple] at (1,1){}; 
  \node[box,fill=mypurple] at (2,1){}; 
  \node[box,fill=mypurple] at (3,1){}; 
  \node[box,fill=mypurple] at (0,0){};  
  \node[box,fill=mypurple] at (1,0){}; 
  \node[box,fill=mypurple] at (2,0){};
  \node[box,fill=mypurple] at (3,0){};  
  
  \node at (1.5, -1) {c}; 

\end{tikzpicture}
\caption{An example move for \ruleset{Transverse Wave}. \textbf{a}: The player chooses column 2, which has a green cell.  \textbf{b}: Indigo cells denote those that will become purple.  These include the previously-green cells in column 2 as well as the green cells in rows where column 2 had purple cells.  \textbf{c}: The new position after all cells are changed to be purple.}
\label{fig:TransverseWaveSampleMove}
\end{figure}
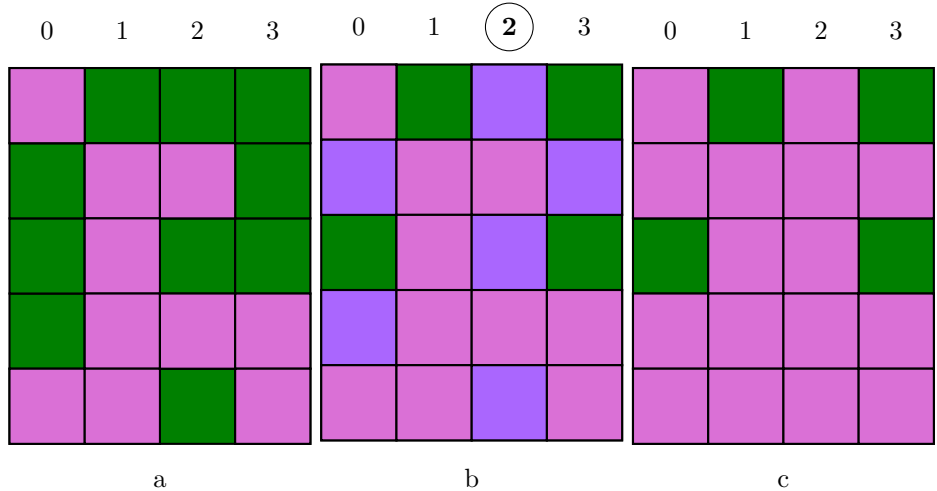

Note that purple cells cannot change to green, and each move increases the number of purple cells in the grid by at least one.
In fact,  each move turns at least one column into one with all purple cells.
Thus, any position with dimension $m$ by $n$ must end in at most $n$ turns, and height of {\sc Transverse Wave}'s game tree is 
at most $n$.
Consequently, {\sc Transverse Wave} is solvable in polynomial space.

Note also that due to the \textit{transverse wave}, the selection of 
column $j$ could make some other feasible columns infeasible.
In Section \ref{Sec:GraphStructure},
  we will show that the interaction among 
  columns introduce sufficiently rich mathematical structures for
{\sc Transverse Wave} to efficiently encode 
  any PSPACE-complete games such as 
  {\sc Hex}
\cite{Gale:1979,NashHex,EvenTarjanHex, Reisch:1981},
  {\sc Avoid True} \cite{DBLP:journals/jcss/Schaefer78},  {\sc Node Kalyes} \cite{DBLP:journals/jcss/Schaefer78}, {\sc Go} \cite{LichtensteinSipser:1980}, {\sc Geography}
\cite{LichtensteinSipser:1980}.
In other words, {\sc Transverse Wave} is a PSPACE-complete impartial game.  

We have implemented {\sc Transverse Wave} 
in HTML/Javascript.\footnote{Web version: \url{https://turing.plymouth.edu/~kgb1013/DB/combGames/transverseWave.html}}

\subsection{Dual Logical Interpretations of {\sc Transverse Wave}}

{\sc Transverse Wave} uses only two colors; a position can be expressed naturally with a Boolean matrix.
Furthermore, making a move can 
  be neatly captured by basic Boolean functions.
Let consider the following two combinatorial games over Boolean matrices that are isomorphic to {\sc Transverse Wave}. 
While these logic associations are straightforward, 
they set up stimulating connections to combinatorial games inspired by social influence  and quantum superposition.

We use the following standard notation for matrices: 
For an $m\times n$ matrix $\mathbf{A}$, $i\in [m]$ and $j\in [n]$, let  
$\mathbf{A}[i,:]$, $\mathbf{A}[:,j]$, $\mathbf{A}[i,j]$ denote, respectively, the $i^{th}$ row,  $j^{th}$ column, and the $(i,j)^{th}$ entry in $\mathbf{A}$.

\begin{gameruleset}[{\sc Crosswise AND}]
For integers $m, n > 0$, {\sc Crosswise AND} plays on an $m \times n$ Boolean matrix $\mathbf{B}$.

During the game, two players alternatively select $j \in [n]$, where $j$ is feasible for $\mathbf{B}$ if  $\mathbf{B}[:,j] \neq \vec{0}$. 
The move with selection $j$ then changes the 
  Boolean matrix to one as the following: $\forall i \neq j \in [m]$, 
the $i^{th}$ row takes a component-wise AND with its $j^{th}$ bit, $\mathbf{B}[i,j]$, then the $(i,j)^{th}$ entry is set to 0. Under normal play, the player with no feasible column to choose loses the game.
\end{gameruleset}

By mapping purple cell to Boolean 0 (i.e., \texttt{false})
  and green cell to Boolean 1 (i.e.,\texttt{true}), 
  and purple transverse wave to crosswise-AND-with-zero, we have:
  
\begin{proposition}
{\sc Transverse Wave} and {\sc Crosswise AND} are isomorphic games.
\end{proposition}

\begin{gameruleset}[{\sc Crosswise OR}]
For integer parameters $m, n > 0$, {\sc Crosswise OR} plays on an $m \times n$ Boolean matrix $\mathbf{B}$.

During the game, two players alternatively 
  select $j \in [n]$, where $j$ is feasible for $\mathbf{B}$ if  
  $\mathbf{B}[:,j] \neq \vec{1}$. 
The move with selection $j$ then changes the Boolean matrix 
  to one as the following: $\forall i \neq j \in [m]$, 
the $i^{th}$ row takes a component-wise OR with its $j^{th}$ bit, 
$\mathbf{B}[i,j]$, then the $(i,j)^{th}$ entry is set to 1. 
Under normal play, the player with no feasible column to choose loses the game.
\end{gameruleset}

By mapping purple cell to Boolean 1,
  and green cell to Boolean 0, and purple transverse wave to crosswise-OR-with-one, we have:
  
\begin{proposition}
{\sc Transverse Wave} and {\sc Crosswise OR} are isomorphic games.
\end{proposition}

\subsection{{\sc Transverse Wave} Game Values}
\label{CGT Values}

As we will discuss later, in general, {\sc Transverse Wave} is PSPACE-complete. 
So, we have no hope for an efficient complete characterization for 
the game values. 
In this subsection, we show that the game values for a specific case of \ruleset{Transverse Wave}, where each column with has no more than 1 purple tile, 
can be fully characterized.  Because we are able to classify the moves into two types of options, and by extension can define the game by two parameters, a fun and interesting nimber version of Pascal-Like triangle arises.

\begin{theorem}[Pascal-Like Nimber Triangle]
Let $p$ be the number of rows with at least 1 purple tile and $k$ be the number of rows with an odd number of purple tiles, and $q = 0$ if there are an even number of columns with only green tiles, and 1 otherwise.  We define $G'$ to be

$G' = \begin{cases}
0, & (k \text{ is even and } p > 2k) \text{ or } (k \text{ is odd and }p < 2k)\\
*, & (k \text{ is even and } p < 2k) \text{ or } (k \text{ is odd and }p > 2k)\\
*2, & p = 2k
\end{cases}$ 

If $G$ is a \ruleset{Transverse Wave} position where for every column we can select, there is no more than one purple tile (discounting rows with only purple tiles), then, $G = G' + \ast q$
\end{theorem}

\begin{proof}
We will call the rows with an odd number of purple tiles the \textit{odd parity} rows and the ones with an even number of purple tiles the \textit{even parity} rows. We claim that $G'$ is the game value without including the columns with only green tiles as options, and $G$ to be the game value with them.

Assuming $G'$ is as we claim, it isn't difficult to see that $G$ is correct. For each all-green column, they can be selected at any time, and don't effect the game at all (since they can always be selected). Thus, each all-green column is just a $\ast$, so $G = G'$ if there are an even number of all-green columns and $G = G' + \ast$ if there are an odd number of all-green columns.

We have an illustrative triangle of cases of $G'$ of up to 8 rows in Figure \ref{fig:AvoidTrueTriangle}. 

\begin{figure}[h!]

\newcommand{\pascalAngle}{0}

\begin{center}\begin{tikzpicture}[node distance = .05cm]

  \node (row8) at (0,0) {$8$};
  \node (row7) at (0,.5) {$7$};
  \node (row6) at (0,1) {$6$};
  \node (row5) at (0,1.5) {$5$};
  \node (row4) at (0,2) {$4$};
  \node (row3) at (0,2.5) {$3$};
  \node (row2) at (0, 3) {$2$};
  \node (row1) at (0, 3.5) {$1$};
  \node at (0, 4) {\# rows};
  
  \node at (4, 5.5) {\# odds};

  \node[rotate=\pascalAngle] (odds0) at (4.5, 5) {0};
  \node[rotate=\pascalAngle] (odds1) at (4.75, 4.5) {1};
  \node[rotate=\pascalAngle] (odds2) at (5, 4) {2};
  \node[rotate=\pascalAngle] (odds3) at (5.25, 3.5) {3};
  \node[rotate=\pascalAngle] (odds4) at (5.5, 3) {4};
  \node[rotate=\pascalAngle] (odds5) at (5.75, 2.5) {5};
  \node[rotate=\pascalAngle] (odds6) at (6, 2) {6};
  \node[rotate=\pascalAngle] (odds7) at (6.25, 1.5) {7};
  \node[rotate=\pascalAngle] (odds8) at (6.5, 1) {8};
  
  \node (row1A) at (3.75, 3.5) {$0$};
  \node (odds0Start) at (3.9, 3.8) {};
  \node (row1Start) [left=of row1A] {};
  \node  at (4.25, 3.5) {$0$};
  \node (odds1Start) at (4.4, 3.8) {}; 
  
  \node (row2A) at (3.5, 3) {$0$};
  \node (row2Start) [left=of row2A] {};
  \node at (4, 3) {$*2$};
  \node at (4.5, 3) {$*$};
  \node (odds2Start) at (4.65, 3.3) {};
  
  \node (row3A) at (3.25, 2.5) {$0$};
  \node (row3Start) [left=of row3A] {};
  \node at (3.75, 2.5) {$*$};
  \node at (4.25, 2.5) {$*$};
  \node at (4.75, 2.5) {$0$};
  \node (odds3Start) at (4.9, 2.8) {};
  
  \node (row4A) at (3, 2) {$0$};
  \node (row4Start) [left=of row4A] {};
  \node at (3.5, 2) {$*$};
  \node at (4, 2) {$*2$};
  \node at (4.5, 2) {$0$};
  \node at (5, 2) {$*$};
  \node (odds4Start) at (5.15, 2.3) {};
  
  \node (row5A) at (2.75, 1.5) {$0$};
  \node (row5Start) [left=of row5A] {};
  \node at (3.25, 1.5) {$*$};
  \node at (3.75, 1.5) {$0$};
  \node at (4.25, 1.5) {$0$};
  \node at (4.75, 1.5) {$*$};
  \node at (5.25, 1.5) {$0$};
  \node (odds5Start) at (5.4, 1.8) {};
  
  \node (row6A) at (2.5, 1) {$0$};
  \node (row6Start) [left=of row6A] {};
  \node at (3, 1) {$*$};
  \node at (3.5, 1) {$0$};
  \node at (4, 1) {$*2$};
  \node at (4.5, 1) {$*$};
  \node at (5, 1) {$0$};
  \node at (5.5, 1) {$*$};
  \node (odds6Start) at (5.65, 1.3) {};
  
  \node (row7A) at (2.25, .5) {$0$};
  \node (row7Start) [left=of row7A] {};
  \node at (2.75, .5) {$*$};
  \node at (3.25, .5) {$0$};
  \node at (3.75, .5) {$*$};
  \node at (4.25, .5) {$*$};
  \node at (4.75, .5) {$0$};
  \node at (5.25, .5) {$*$};
  \node at (5.75, .5) {$0$};
  \node (odds7Start) at (5.90, .8) {};
  
  \node (row8A) at (2, 0) {$0$};
  \node (row8Start) [left=of row8A] {};
  \node at (2.5, 0) {$*$};
  \node at (3, 0) {$0$};
  \node at (3.5, 0) {$*$};
  \node at (4, 0) {$*2$};
  \node at (4.5, 0) {$0$};
  \node at (5, 0) {$*$};
  \node at (5.5, 0) {$0$};
  \node at (6, 0) {$*$};
  \node (odds8Start) at (6.15, .3) {};
  
  \draw[->] 
    (row1) edge (row1Start)
    (row2) edge (row2Start)
    (row3) edge (row3Start)
    (row4) edge (row4Start)
    (row5) edge (row5Start)
    (row6) edge (row6Start)
    (row7) edge (row7Start)
    (row8) edge (row8Start)
    
    (odds0) edge (odds0Start)
    (odds1) edge (odds1Start)
    (odds2) edge (odds2Start)
    (odds3) edge (odds3Start)
    (odds4) edge (odds4Start)
    (odds5) edge (odds5Start)
    (odds6) edge (odds6Start)
    (odds7) edge (odds7Start)
    (odds8) edge (odds8Start)
    ;

\end{tikzpicture}\end{center}
\caption{A Pascal-Like Nimber Triangle of the values of \ruleset{Transverse Wave} positions with no overlapping literals.}
\label{fig:AvoidTrueTriangle}
\end{figure}

If the player chooses a column where the purple tile is in an odd parity row, then an even number of other rows share that single purple cell. Later selecting any of those rows makes no other change to the game state, so they can each be considered to contribute (additively) a $\ast$, for a total of $2k\ast = 0$. Then the resulting option's value is just the same as one with $p-1$ rows and $k-1$ rows with odd parity. This just the value above and left in the triangle previously referenced.

If the player chooses a column where the purple tile is in a row with even parity, then an even number of other rows also have that single purple cell.  Thus, the result is an odd number of $\ast$ is added $(2k+1)\ast = \ast$ added to the option.  Thus, the resulting game is the value above and right in the table (the same number of odd rows and one less row overall) plus $\ast$.

\begin{table}[h!]
\begin{tabular}{r | c | c | l}
  Case & Above left & Above right & Value\\ \hline
  a & & &  0\\ \hline
  b & & 0 &  0\\ \hline
  c & 0 & & $*$ \\ \hline
  d & 0 & 0 & $*2$ \\ \hline
  e & 0 & $*$ & $*$ \\ \hline
  f & 0 & $*2$ & $*$ \\ \hline
  g & $*$ & & 0 \\ \hline 
  h & $*$ & 0 & 0 \\ \hline
  i & $*$ & $*$ & $*2$ \\ \hline
  j & $*$ & $*2$ & 0 \\ \hline
  k & $*2$ & 0 & 0 \\ \hline 
  l & $*2$ & $*$ & $*$ \\ \hline
\end{tabular}
\caption{Values for a game with a given position based on what is above and left in the triangle and what is above and right}
\label{table:TriangleTable}
\end{table}

By inspection, note that Table \ref{table:TriangleTable} represents the correct value for each possible parents in the game tree. 

Now, we have 5 cases which invoke those table cases.
\begin{enumerate}
  \item $k$ even and $p > 2k$: case a, b, h, or j (thus value 0)
  \item $k$ odd and $p < 2k$: case a, g, h, or k (thus value 0)
  \item $k$ even and $p < 2k$: case c, e, or l (thus value $\ast$)
  \item $k$ odd and $p > 2k$: case e or f (thus value $\ast$)
  \item $p = 2k$: case d or i (thus in $\ast 2$)
\end{enumerate}

Let's prove the correctness of these cases. Clearly when the game is a single row it holds by inspection.

We can assume by induction that it holds row $p - 1$.

Now, let's show that it holds, for $k$ even and $p > 2k$. This is either it is the base case (case $a$), because there are 0 odd rows (case $b$), or $k$ some other even number less than $\frac{1}{2}p$. It's left parent will have $k - 1$ with $p - 1 > 2(k - 1)$, which, by induction, must be $\ast$. It's right parent has $p - 1 \geq 2k$, and thus is either 0 or $\ast 2$. Thus, it must be either case $h$ or $j$.

Now, let's look at $k$ odd and $p < 2k$. Then, either it is the base case (case $a$), $k = p$ (then it must be case $g$, since it only has a single left parent with even $k$), or it is some other odd $k$ such that $p < 2k$. Then, it has a right parent with odd $k$ and $p < 2k$, which is 0. And the left parent has even $k$ and $p - 1 \leq 2(k - 1)$, thus a left parent of either $\ast$ or $\ast 2$, which is case $h$ and $k$, respectively.

If $k$ is even and $p < 2k$, then either $k = p$, in which case it has a single left parent which is in case 2, which is case $c$, or it is some other even $k$ with $p < 2k$. In that case, the top right parent will have the same $k$ and thus be case 3 and thus $\ast$. The top left parent will have $k - 1$ and $n - 1 \leq 2(k-1)$, and thus be 0 or $\ast 2$. This is $e$ or $l$.

If $k$ is odd and $p - 1 > 2(k - 1)$, then we know the top left parent has $p > 2(k - 1)$ which is 0 (since it is case 1). The top right parent has $p - 1 \geq 2k$, and is thus 0 or $\ast 2$, and is thus $e$ or $f$.

Finally, if $p = 2k$, then $k$ can either be odd or even. If it is odd, then the left parent is has even $k$ and $p-1 > 2(k-1)$ and is thus 0, and the right parent has even $k$ and $p < 2k$ and is thus 0, putting this in case $d$. If $k$ is even, then the left parent is odd and thus $\ast$, and the right parent is even and thus $\ast$, putting us in case $i$.
\end{proof}

We also have some game miscellaneous game values, and have examples of up to $\ast 7$, as shown in in Table \ref{table:values}.

\begin{table}[h!]
\begin{tabular}{r | c | c | l}
  Nimber & Rows (shorthand) & Other Columns \\ \hline
  0 & (0) & \\ \hline
  $\ast$ & (0) (01) & \\ \hline
  $\ast 2$ & (01) (2) & \\ \hline
  $\ast 3$ & (01) (2) & 3 \\ \hline
  $\ast 4$ & (01) (234) (035) & \\ \hline
  $\ast 5$ & (01) (234) (035) & 6 \\ \hline
  $\ast 6$ & (012) (034) (0156) (2578) & \\ \hline 
  $\ast 7$ & (012) (034) (0156) (2578) & 9\\ \hline
\end{tabular}
\caption{Instances of values up to $\ast 7$. The shorthand uses parenthesis to indicated rows and numbers to indicate purple columns So, for example, the $\ast 3$ case the first row with columns 0 and 1 colored purple, row two with column 2 colored purple, and with another column 3 which isn't colored purple in any row}
\label{table:values}
\end{table}

Note that all of these CGT results can be extended to the related games that are exact embeddings of this, as we will discuss later, of \ruleset{Avoid True} and \ruleset{Demi-Quantum Boolean Nim}.

Although we can only characterize \ruleset{Transverse Wave} in very special cases, 
the Pascal-like formation of these game values provides us a glimpse of 
a potential elegant structure.
Something interesting to explore in the future is whether one can cleanly characterize the game values when the game is restricted to have only two purple tiles in each column. 
And then if so, one would like to see how large of the parameter one can characterize
and when we encounter intractability.
Answering these questions can also tell us about the values for the other related games as well.

\section{Connection to a Social-Influence-Inspired Combinatorial Game}

Because mathematical principles are ubiquitous, combinatorial game theory is a field intersecting many disciplines. Combinatorial games have drawn inspiration widely from logic to topology, from military combat to social sciences, from graph theory to game theory. Because the field cherishes challenging games with simple rulesets and elegant game boards, combinatorial game design is also a distillation process, aiming to derive elementary moves and transitions in order capture the essence of complex phenomena that inspire the game designers.

In this and the next sections, we discuss two games whose intersection contains {\sc Transverse Wave}.
While they have found a common ground, these games rooted from different research fields. 
The first game, called {\sc Friend Circle} is 
  motivated by viral marketing
\cite{RichardsonDomingos,KKT,CTZGraphBasisSocialInfluence}, while the second, called {\sc Demi-Quantum Nim}, was inspired by quantum superpositions \cite{goff2006quantum,dorbec2017toward,BurkeFerlandTengQCGT}.
In this Section, we first focus on {\sc Friend Circle}.

\newcommand{\fcTrue}{\text{\textbf{t}}}
\newcommand{\fcFalse}{\text{\textbf{f}}}

\subsection{Viral-Marketing Broadcasting: {\sc Friend Circle}}

In many ways, viral marketing itself is a game of financial optimization.
It aims to convert more people, through network-propagation-based social influence, by strategically investing in a seed group of people \cite{RichardsonDomingos,KKT}.
In the following combinatorial game inspired by viral marketing,
 {\sc Friend Circle}, we use a high-level perspective 
 of social influence and social networks. 
Consider a social-network universe like Facebook, 
  where people have their ``circles'' of friends.
They can {\em broadcast} to all people in their friend circles (with) a single post), or they can individually interact with some of
their friends (via various personalized means).
We will use individual interaction to set up the game position.
In  game {\sc Friend Circle}, only broadcast-type of 
  interaction is exploited.  
We will use the following traditional graph-theory notation:
In an undirected graph $G=(V,E)$, for each $v\in V$, 
  the neighborhood of
$v$ in $G$ is $N_G(v) = \{u\ |\ (u,v) \in E\}$.

\begin{gameruleset}[{\sc Friend Circle}]
For a ground set $V = [n]$ (of $n$ people), a \ruleset{Friend Circle} position is defined by a triple $(G, S, w)$, where
\begin{itemize}
  \item $G = (V, E)$ is an undirected graph.  An edge between two vertices represents a friendship between those people.
  \item $S \subset V$ denotes the seed set, and
  \item $w : E \rightarrow \{\fcFalse, \fcTrue \}$ (false and true) represents whether those friends have already spoken about the target product (with at least one recommending it to the other).
\end{itemize}

To choose their move, a player--a viral marketing agent--picks a person from the seed set, $v \in S$, such that $\exists\ e = (v, x) \in E$ where $w(e) = \fcFalse$.  This represents choosing someone who hasn't spoken about the product to at least one of their friends.

The result of this move is a new position $(G, S, w')$, where $w'$ is the same as $w$ except that $\forall x \in N_G(v):$
\begin{itemize}
    \item $w'((v, x)) = \fcTrue$, and
    \item if $w((v, x)) = \fcTrue$ then $\forall y \in N_G(x): w'((x, y)) = \fcTrue$.
\end{itemize}



\end{gameruleset}

An example of a \ruleset{Friend Circle} move is shown in Figure \ref{fig:friendCircle}.

\newcommand{\trueEdge}{$t$}
\newcommand{\falseEdge}{$f$}
\newcommand{\newTrueEdge}{$\mathbf{t}$}
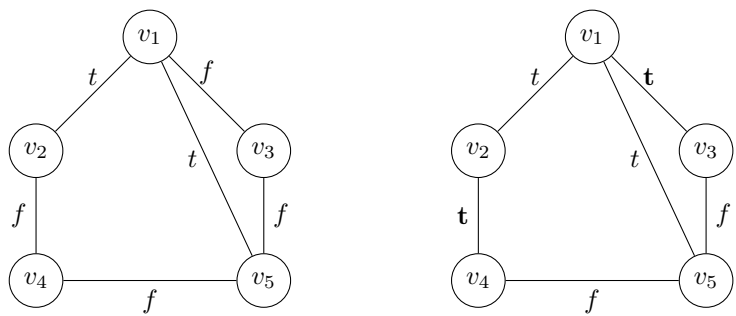
\begin{figure}[h!]
\begin{center}\begin{tikzpicture}

  \node[draw, circle] (v1) {$v_1$};
  \node[draw, circle] (v2) [below left=of v1] {$v_2$};
  \node[draw, circle] (v3) [below right=of v1] {$v_3$};
  \node[draw, circle] (v4) [below=of v2] {$v_4$};
  \node[draw, circle] (v5) [below=of v3] {$v_5$};
  
  \draw[-]
    (v1) edge node [above] {\trueEdge} (v2)
    (v1) edge node [above] {\falseEdge} (v3)
    (v1) edge node [left] {\trueEdge} (v5)
    (v2) edge node [left] {\falseEdge} (v4)
    (v3) edge node [right] {\falseEdge} (v5)
    (v4) edge node [below] {\falseEdge} (v5)
    ;

\end{tikzpicture}\hspace{2cm}\begin{tikzpicture}

  \node[draw, circle] (v1) {$v_1$};
  \node[draw, circle] (v2) [below left=of v1] {$v_2$};
  \node[draw, circle] (v3) [below right=of v1] {$v_3$};
  \node[draw, circle] (v4) [below=of v2] {$v_4$};
  \node[draw, circle] (v5) [below=of v3] {$v_5$};
  
  \draw[-]
    (v1) edge node [above] {\trueEdge} (v2)
    (v1) edge node [above] {\newTrueEdge} (v3)
    (v1) edge node [left] {\trueEdge} (v5)
    (v2) edge node [left] {\newTrueEdge} (v4)
    (v3) edge node [right] {\falseEdge} (v5)
    (v4) edge node [below] {\falseEdge} (v5)
    ;

\end{tikzpicture}\end{center}
\label{fig:friendCircle}
\caption{Example of a \ruleset{Friend Circle} move.  In the position on the left, let the seed set $S = \{v_1, v_2, v_3, v_4\}$, all of which are acceptable to choose because they all have an incident false edge.  If a player chooses $v_2$, then the result is the right-hand position.  In the second position, $v_2$ has had all of it's incident edges become true.  In addition, since $(v_1, v_2)$ was true, all of $v_1$'s incident edges have also changed to true.  The altered edges in the figure are represented in bold.  Note that in the resulting position, the next player can only choose to play at either $v_3$ and $v_4$, as $v_1$ and $v_2$ have only true edges and $v_5 \notin S$.}
\end{figure}

By inducing people in the seed's friend circle who had existing 
intersection with the chosen seed to broadcast, 
{\sc Friend Circle} emulates an elementary two-step network cascading in
social influence.

\subsection{Intractability of {\sc Friend Circle}}

\label{Sec:NKFC}

We first connect {\sc Friend Circle} to the classical graph-theory game {\sc Node-Kayles}.

\begin{gameruleset}[(\sc Node-Kayles)]
The starting position of {\sc Node-Kayles} 
  is an undirected graph $G = (V,E)$.

During the game, two players alternate turns selecting vertices,
where a vertex $v\in V$ is feasible if neither it has already been selected in 
the previous turns nor any of its neighbors has already been selected.
The player who has no more feasible move loses the game.
\end{gameruleset}

When the selected vertices form an {\em maximal independent set}
of $G$, the next player cannot make a move, and hence loses
the game.
It is well-known that {\sc Node-Kayles} is \cclass{PSPACE} complete \cite{DBLP:journals/jcss/Schaefer78}.

\begin{theorem}[{\sc Friend Circle} is \cclass{PSPACE}-complete]
The problem of determining whether a {\sc Friend Circle} position is winnable is \cclass{PSPACE}-complete. 
\end{theorem}
\begin{proof}
First, we show that {\sc Friend Circle} is \cclass{PSPACE}-solvable.
During a game of {\sc Friend Circle} starting at $(G, S, w)$, once a node $s\in S$ is selected by one of the players, $s$ all edges incident to $s$ become \fcTrue.  Since true edges can never later become \fcFalse, $s$ can never again be chosen for a move and the height of the game tree is at most $|S|$. 
Then, by the standard depth-first-search (DFS) procedure for evaluating the game tree for $(G,S, w)$ in {\sc Friend Circle}, we can determine the outcome class of in polynomial space.

To establish that {\sc Friend Circle} is a PSPACE-hard game, we reduce  {\sc Node-Kayles} to {\sc Friend Circle}. 


\begin{figure}[h!]
\begin{center}\begin{tikzpicture}[node distance = 1cm]
    \node[draw, circle] (v) {$v$};
    
    \node[draw, circle] (n1) [above right=of v] {};
    \node[draw, circle] (n2) [above= of v] {};
    \node[draw, circle] (n3) [below left=of n2] {};
    \node[draw, circle] (n5) [below=of v, label=below:{$v$ and $N_{G_0}(v)$}] {};
    \node[draw, circle] (n4) [below=of n3] {};
    
    \draw[-]
        (v) edge (n1)
        (v) edge (n2)
        (v) edge (n3)
        (v) edge (n4)
        (v) edge (n5);
\end{tikzpicture}
\hspace{2cm}
\begin{tikzpicture}[node distance = 1cm]
    \node[draw, circle] (v) {$v$};
    
    \node[draw, circle] (n1) [above right=of v] {};
    \node[draw, circle] (n2) [above= of v] {};
    \node[draw, circle] (n3) [below left=of n2] {};
    \node[draw, circle] (n5) [below=of v] {};
    \node[draw, circle] (n4) [below=of n3] {};
    
        
    \node[draw, circle] (t) [right=of v] {$t_v$};
    
    \draw[-]
        (v) edge node [above] {\fcTrue} (n1)
        (v) edge node [left] {\fcTrue} (n2)
        (v) edge node [above] {\fcTrue} (n3)
        (v) edge node [above] {\fcTrue} (n4)
        (v) edge node [left] {\fcTrue} (n5)
        (v) edge node [below] {\fcFalse} (t);
        
\end{tikzpicture}\end{center}
\label{fig:NodeKaylesToFriendCircle}
\caption{Example of the reduction from \ruleset{Node Kayles} to \ruleset{Friend Circle}.  On the left is a \ruleset{Node Kayles} vertex and it's neighborhood.  On the right is those same vertices, along with $t_v$ with \fcTrue weights on all the old edges and \fcFalse on the new edge with $(v, t_v)$.}
\end{figure}
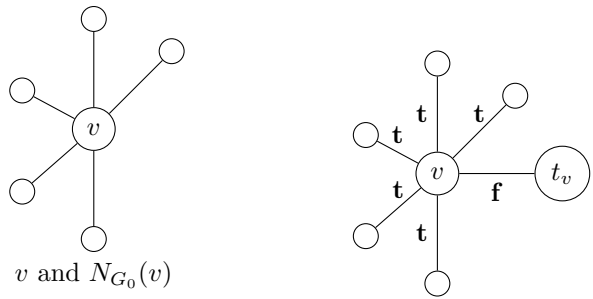

Suppose we have a {\sc Node-Kayles} instance at graph $G_0 = (V_0, E_0)$.
For the reduced {\sc Friend Circle} position, we create a new graph 
$G=(V,E)$ as the following.
First, for each $v\in V_0$, we introduce a new vertex $t_v$.
Let $T_0 = \{t_v| v\in V_0\}$, so $V = V_0 \cup T_0$.  In addition, let $E_1 = \{(v, v_t)\ |\ v \in V_0\}$, so $E = E_0 \cup E_1$.
Next, we set the weights:
\begin{itemize}
  \item $\forall e \in E_0: w(e) = \fcTrue$, and
  \item $\forall e \in E_1: w(e) = \fcFalse$
\end{itemize}
as shown in \ruleset{Friend Circle}.
Last, we set $S = V_0$.

We now prove that {\sc Friend Circle} is winnable at $(G, S, w)$
if and only {\sc Node Kayles} is winnable at $G_0$.
Note that because $\forall v\in V_0$, $w (v,t_v) = \fcFalse$, all vertices in $V_0$ are feasible choices for the current player in 
{\sc Friend Circle}.  As the game progresses, vertices in $V_0$ are no longer able to be chosen when their edge to the $t$ vertex becomes true.
From here the argument is very simple: each play on vertex $v \in V_0$ in \ruleset{Node Kayles} corresponds exactly to the play on $v$ in \ruleset{Friend Circle}.  In \ruleset{Node Kayles}, when $v$ is chosen, itself and all its neighbors, $N_{G_0}(v)$ are removed from future consideration.  In \ruleset{Friend Circle}, $v$ is also removed because the edge $(v, t_v)$ becomes \fcTrue.  In addition, since all neighboring vertices $x \in N_{G_0}(v)$ share a \fcTrue edge with $v$, their edge with $t_x$ will also become \fcTrue, but no other vertices will be removed from future choice.  
\end{proof}

\subsection{{\sc Transverse Wave} in {\sc Friend Circle}}

\label{Sec:TWFC}


We now show that {\sc Friend Circle} \textit{contains} {\sc Transverse Wave} as \textit{special cases}.
In the proposition below and the rest of the paper, we say that two game instances are \textit{isomorphic} to each other if there exists a \textit{bijection} between their moves such that their game trees are isomorphic under this bijection.

\begin{proposition}[Social-Influence Connection of {\sc Transverse Wave}]
\label{prop:TWSocialInfluence}
For any complete bipartite graph $G = (V_1,V_2, E)$ over two disjoint ground sets $V_1$ and $V_2$ (i.e., with $E = V_1 \times V_2$), any weighting $w: E \rightarrow 
\{\fcFalse, \fcTrue\}$, and seeds $S = V_1$, 
{\sc Friend Circle} position $(G, S, w)$ is isomorphic to {\sc Crosswise OR} over a pseudo-adjacency matrix $\mathbf{A}_G$ for $G$ with $V_1$ as columns and $V_2$ as rows.  In this matrix, we will have the entry at column $x \in V_1$ and row $y \in V_2$ be 0 if $w((x,y)) = \fcFalse$ and 1 if the weight is \fcTrue.
\end{proposition}

Note that by varying 
$w: E \rightarrow \{\fcFalse, \fcTrue]\}$, one can realize any $|V_1|\times |V_2|$ Boolean matrix with $\mathbf{A}_G$.
Thus, {\sc Friend Circle} \textit{generalizes} {\sc Transverse Wave}.

\begin{proof}
Imagine these two games are played in tandem.
We map the selection of a vertex $v\in S = V_1$ to the selection of the
column associated with $v$ in the matrix $\mathbf{A}_G$ of $G$.
Because $G$ is a complete bipartite graph, $v$ is feasible for {\sc Friend Circle} if there exists $u\in V_2$ such that $w(u,v) = \fcFalse$.
Thus, the column associated with $v$ in $\mathbf{A}_G$ is not all 1s.
This is precisely the condition that $v$ is feasible in {\sc Crosswise OR} over $\mathbf{A}_G$.
The Direct Influence at $v$ in {\sc Friend Circle} over $G$ changes all $v$'s edges to \fcTrue and the subsequent Cascading Influence
on $v$'s initially \fcTrue neighbors in $V_2$ is isomorphic to crosswise ORs.
Thus, {\sc Friend Circle} on $G$ is isomorphic to {\sc Crosswise OR} over
$\mathbf{A}_G$.
\end{proof}

\begin{figure}[h!]
\begin{center}\begin{tikzpicture}[node distance = .5cm]

  \node[] (leftSide) {$V_1$};
  \node[draw, circle] (1) [below=of leftSide] {$1$};
  \node[draw, circle] (2) [below=of 1] {$2$};
  \node[draw, circle] (3) [below=of 2] {$3$};
  \node[draw, circle] (4) [below=of 3] {$4$};
  \node[draw, circle] (5) [below=of 4] {$5$};
  \node[draw, circle] (6) [below=of 5] {$6$};
  
  \node[] (rightSide) at (3, 0) {$V_2$};
  \node[draw, circle] (a) at (3,-2.2) {$a$};
  \node[draw, circle] (b) [below=of a] {$b$};
  \node[draw, circle] (c) [below=of b] {$c$};
  \node[draw, circle] (d) [below=of c] {$d$};
  
  \draw[-]
    (1) edge (a)
    (1) edge (b)
    (1) edge (d)
    (2) edge (b)
    (2) edge (c)
    (2) edge (d)
    (3) edge (a)
    (3) edge (b)
    (4) edge (c)
    (4) edge (d)
    (5) edge (a)
    (5) edge (c)
    (6) edge (c)
    ;

\end{tikzpicture}\hspace{2cm}
\begin{tikzpicture}[box/.style={rectangle,draw=black,thick, minimum size=1cm}]

  \foreach \x in {1,...,6}{
      \foreach \y in {-1,...,-4}
          \node[box, fill=mygreen] at (\x,\y){};
  }
  
  \node at (1,0) {1};
  \node at (2,0) {2};
  \node at (3,0) {3};
  \node at (4,0) {4};
  \node at (5,0) {5};
  \node at (6,0) {6};
  
  \node at (0,-1) {$a$};
  \node at (0,-2) {$b$};
  \node at (0,-3) {$c$};
  \node at (0,-4) {$d$};
  
  \node[box,fill=mypurple] at (1,-1){}; 
  \node[box,fill=mypurple] at (1,-2){}; 
  \node[box,fill=mypurple] at (1,-4){}; 
  \node[box,fill=mypurple] at (2,-2){}; 
  \node[box,fill=mypurple] at (2,-3){}; 
  \node[box,fill=mypurple] at (2,-4){}; 
  \node[box,fill=mypurple] at (3,-1){}; 
  \node[box,fill=mypurple] at (3,-2){};  
  \node[box,fill=mypurple] at (4,-3){};  
  \node[box,fill=mypurple] at (4,-4){};  
  \node[box,fill=mypurple] at (5,-1){};  
  \node[box,fill=mypurple] at (5,-3){};  
  \node[box,fill=mypurple] at (6,-3){};

\end{tikzpicture}\end{center}
\caption{On the left is a \ruleset{Friend Circle} position on the complete bipartite graph between $V_1$ and $V_2$, where the seed set $S = V_1$.  Instead of labelling edges, we have removed all false edges and are including only true edges.  On the right is the equivalent \ruleset{Transverse Wave} position.  The purple cells correspond to the (true) edges in the bipartite graph.}
\label{fig:transverseWaveFriendCircle}
\end{figure}
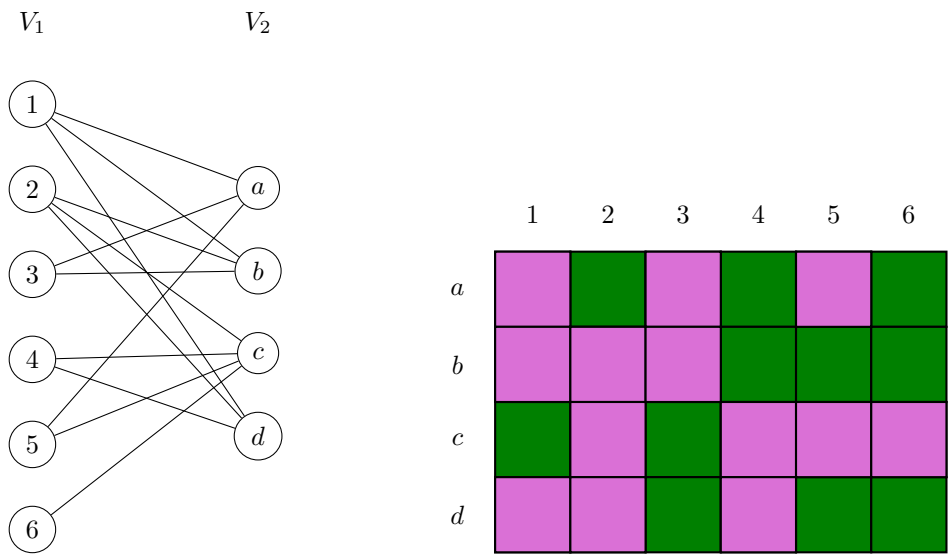


\section{Connection to Quantum Combinatorial Game Theory}
\label{Sec:Quantum}

In this section, we discuss the connection of {\sc Transverse Wave} to 
a basic quantum-inspired combinatorial game.

Quantum computing is inspirational not only because the advances of quantum technologies have the potential to drastically change the landscape of computing and digital security, but also because the quantum framework---powered by superpositions, entanglements, and collapses---has fascinating
mathematical structures and properties.
Not surprisingly, quantumness has already found their way to  enrich combinatorial game theory.
In early 2000s, Allan Goff introduced the basic 
  quantum elements into {\sc Tic-Tac-Toe}, 
  as conceptual illustration of quantum physics \cite{goff2006quantum}.
The quantum-generalization of {\sc Tic-Tac-Toe} expands
 players strategy space with superposition of 
 classical moves, creating game boards with entangled components.
Consistency-based conditions for collapsing can then reduce the degree of 
possible realizations in the potential parallel game  scenarios.
In 2017,  Dorbec and Mhalla \cite{dorbec2017toward} 
 presented a general framework, motivated by Goff's concrete adventures,
 for quantum-inspired extension of  classical combinatoral games.
Their framework enabled our recent work \cite{BurkeFerlandTengQCGT} on the 
 structures and complexity of quantum-inspired games, 
 which also led us to make the connection with
 logic and social-influence inspired {\sc Transverse Wave} and 
 {\sc Friend Circle} in this paper.

\subsection{Superposition of Moves and Game Realizations}

In this subsection, we briefly discuss quantum combinatorial game theory (QCGT) to introduce needed concepts and notations for this paper. More detailed discussions of QCGT can be found in 
\cite{goff2006quantum,dorbec2017toward,BurkeFerlandTengQCGT}.

\begin{itemize}
\item {\bf\em Quantum Moves}: A quantum move is a superposition of two or more distinct classical moves.
The superposition of $w$ classical moves $\sigma_1, ...,\sigma_w$---called a $w$-wide quantum move---is denoted by $\langle \sigma_1\ | ... \ | \ \sigma_w\rangle$.

\item {\bf\em Quantum Game Position}: A quantum position is a superposition of two or more distinct classical game positions.
The superposition of $s$ classical positions $b_1, ...,b_s$---called an $s$-wide quantum superposition--is denoted by $\mathbb{B} = \langle b_1\ | ... \ | \  b_s\rangle$.
We call  $b_1, ...,b_s$ {\em realizations} of $\mathbb{B}$. 

We sometimes refer to classical moves and positions as 1-wide superpositions.
\end{itemize}

Classical/quantum moves can be applied to classical/quantum positions.
Variants of Dorbec-Mhalla framework differ in 
the condition when classical moves are 
  allowed to engage with quantum positions.
In this paper, we will focus on the least restrictive flavor---referred to by variant $D$ in \cite{dorbec2017toward,BurkeFerlandTengQCGT}---in which 
moves, classical or quantum, are allowed to interact with game positions, classical or quantum, provided that they are feasible.
There are some subtle differences between these variants, and we direct interested readers to \cite{dorbec2017toward,BurkeFerlandTengQCGT}. 
In this least restrictive flavor, a superposition of moves (including 1-wide superpositions) is {\em feasible} for a quantum position (including 1-wide superpositions) if each move is feasible for some realization in the quantum position.
A superposition of moves and a quantum position of realizations creates a "tensor" of classical interactions in which infeasible classical interactions introduce collapses in realizations.

Quantum moves can have impact on the outcome class of games, even on classical positions. We borrow an illustration (see Figure \ref{fig:nim22-intro}) from \cite{BurkeFerlandTengQCGT} 
showing that quantum moves can change the outcome class of a basic \ruleset{Nim} position.  $(2,2)$ becomes a fuzzy ($\outcomeClass{N}$, a first-player win) position instead of a zero ($\outcomeClass{P}$, a second-player win) position.
Quantumness matters for many combinatorial games, as investigated in
\cite{BurkeFerlandTengQCGT}.

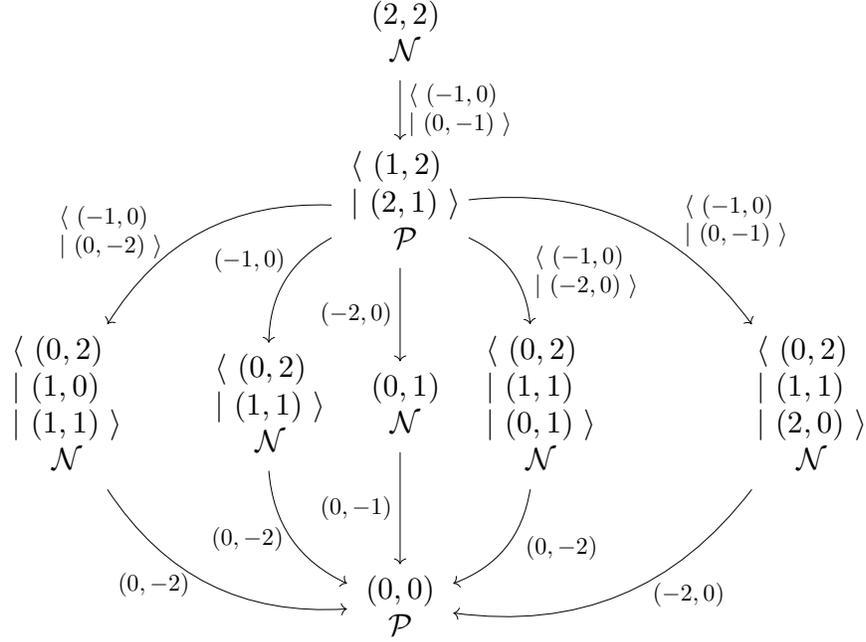
\begin{figure}[ht]
\begin{center}
  \def\scaleAmt{1.3}
  \scalebox{.9}{
  \begin{tikzpicture} [node distance = 1cm]
    \node (start) at (0, .5)  {\scalebox{\scaleAmt}{
        \begin{tabular}{@{}c@{}}
            $(2,2)$ \\ 
            \multicolumn{1}{c}{$\outcomeClass{N}$}
        \end{tabular}}};
    \node (first) at (0, -2) {\scalebox{\scaleAmt}{
        \begin{tabular}{@{}l@{}c@{}}
            $\langle\ (1,2)$\\ $|\ (2,1)\ \rangle$ \\ 
            \multicolumn{1}{c}{$\outcomeClass{P}$}
        \end{tabular}}};
    \node (secondA) at (-5, -5) {\scalebox{\scaleAmt}{
        \begin{tabular}{@{}l@{}l@{}c@{}}
            $\langle\ (0,2)$\\ $ |\ (1, 0)$\\ $|\ (1,1)\ \rangle$ \\ 
            \multicolumn{1}{c}{$\outcomeClass{N}$}
        \end{tabular}}};
    \node (secondB) at (-2, -5) {\scalebox{\scaleAmt}{
        \begin{tabular}{@{}l@{}c@{}}
            $\langle\ (0,2)$\\ $|\ (1,1)\ \rangle$\\ 
            \multicolumn{1}{c}{$\outcomeClass{N}$}
        \end{tabular}}};
    \node (secondC) at (0, -5) {\scalebox{\scaleAmt}{
        \begin{tabular}{@{}c@{}}
            $(0, 1)$ \\ $\outcomeClass{N}$
        \end{tabular}}};
    \node (secondD) at (2, -5) {\scalebox{\scaleAmt}{
        \begin{tabular}{@{}l@{}l@{}c@{}} 
            $\langle\ (0,2)$\\ $|\ (1,1)$\\ $|\ (0,1)\ \rangle$ \\ 
            \multicolumn{1}{c}{$\outcomeClass{N}$}
        \end{tabular}}};
    \node (secondE) at (6, -5) {\scalebox{\scaleAmt}{
        \begin{tabular}{@{}l@{}l@{}c@{}} 
            $\langle\ (0,2)$\\ $|\ (1,1)$\\ $|\ (2,0)\ \rangle$ \\ 
            \multicolumn{1}{c}{$\outcomeClass{N}$}
        \end{tabular}}};
    \node (zero) at (0, -8) {\scalebox{\scaleAmt}{
        \begin{tabular}{@{}c@{}}
            $(0,0)$\\ $\outcomeClass{P}$
        \end{tabular} } };

    \path[->] 
        (start) edge [] node [right, text width=1cm] {$\langle\ (-1, 0)$\\ $|\ (0, -1)\ \rangle$} (first)
        (first) edge [bend right] node [left, text width = 1cm, xshift = -28pt] {$\langle\ (-1, 0)$\\ $|\ (0, -2)\ \rangle$} (secondA)
        (first) edge [bend right] node [above, xshift = -15pt] {$(-1, 0)$} (secondB)
        (first) edge [] node [left] {$(-2, 0)$} (secondC)
        (first) edge [bend left] node [right, text width = 1cm, xshift = 6pt] {$\langle\ (-1, 0)$\\ $|\ (-2, 0)\ \rangle$} (secondD)
        (first) edge [bend left] node [right, xshift = 20pt, text width = 1cm] {$\langle\ (-1,0)$\\ $|\ (0, -1)\ \rangle$} (secondE)
        (secondA) edge [bend right] node [left, xshift = -5pt] {$(0,-2)$} (zero)
        (secondB) edge [bend right] node [left] {$(0,-2)$} (zero)
        (secondC) edge [] node [left] {$(0, -1)$} (zero)
        (secondD) edge [bend left] node [right, xshift = 5pt] {$(0, -2)$} (zero)
        (secondE) edge [bend left] node [right, xshift = 10pt] {$(-2, 0)$} (zero)
        ;
    
  \end{tikzpicture}}
\end{center}  
  \caption{Illustration from \cite{BurkeFerlandTengQCGT}: Winning strategy for Next player in  \ruleset{Quantum Nim} $(2,2)$, showing that quantum moves impact the game's outcome.  (There are four additional move options from $\braket{(1,2)\ |\ (2,1)}$ that are not shown because they are symmetric to moves given.)}
  \label{fig:nim22-intro}
\end{figure}

\begin{gameruleset}[(\sc Nim)]
  A \ruleset{Nim} position is described by a non-negative integer vector, e.g. $(3, 5, 7) = G$, representing heaps of objects (here pebbles).  A turn consists of removing pebbles from exactly one of those heaps.  We describe these moves as a non-positive vector with exactly one non-zero element, e.g. $(0, -2, 0)$.  Each move cannot remove more pebbles from a heap than already exist there.  Thus, the move $(-7, 0, 0)$ is not a legal move from $G$, above.  When all heaps are zero, there are no legal moves.
\end{gameruleset}

To readers familiar with combinatorial game theory, it may seem odd that we explicitly define the description of moves in the game. However, it is integral for playing quantum combinatorial games, as move description effects quantum collapse. For more information, see \cite{BurkeFerlandTengQCGT}.

Quantum interactions between moves and positions,
  as demonstrated in \cite{goff2006quantum,dorbec2017toward},
 can have highly non-trivial impact to Nimber Arithmetic.
In additon, as shown in \cite{BurkeFerlandTengQCGT}, quantum moves can also fundamentally impact the complexity of combinatorial games.

\subsection{{\sc Demi-Quantum Nim}: Superposition of {\sc Nim} Positions}

\label{Sec:DQNTW}
The combinatorial game that contains {\sc Transverse Wave} as a special case is derived from {\sc Nim} \cite{Bouton:1901,Gale:1974}, in a framework motivated by practical implementation of quantum combinatorial games \cite{BurkeFerlandTengQCGT}.

For integer $s>1$, then a $s$-wide quantum {\sc Nim} position of $n$ heaps can be  specified by an $s\times n$ integer matrix, where each row defines a realization of {\sc Nim} position.
For example, the
  $4$-wide quantum {\sc Nim} position with $6$ piles,
\begin{eqnarray*}
\langle (5,3,0,4,2,2)\ |\ (1,3,3,2,1,0) \ | \ (0, 0, 4, 6, 5, 7) \ |\ (4,2,5,0,1,2)  \rangle,
\end{eqnarray*}
can be expressed in the matrix form as:
\begin{eqnarray}
\left( 
\begin{array}{ccccccc}
5 & 3 & 0 & 4 & 2 & 2\\
1 & 3 & 3 & 2 & 1 & 0\\
0 & 0 & 4 & 6 & 5 & 7\\
4 & 2 & 5 & 0 & 1 & 2
\end{array}
\right)
\label{eqn:QuantumNimExample}
\end{eqnarray}

Like the quantum generalization of combinatorial games,
this {\em demi-quantum generalization}
systematically extends any combinatorial 
game by expanding its game positions \cite{BurkeFerlandTengQCGT}.  The intuitive difference here is that players may not introduce new quantum moves, they may only make classical moves, which apply to all (and may collapse some) of the realizations in the current superposition.

\begin{definition}[Demi-Quantum Generalization of Combinatorial Games]
For any game ruleset {\sc R}, 
the demi-quantum generalization of {\sc R}, denoted by, 
{\sc Demi-Quantum-R}, is a combinatorial game 
defined by the interaction of 
classical moves of {\sc R} with quantum positions in {\sc R}.

Central to the demi-quantum transition is the rule for {\em collapses}.
Given a quantum superposition $\mathbb{B}$ and a classical move $\sigma$ of {\sc R},  $\sigma$ is feasible if it is feasible for at least one realization in $\mathbb{B}$, and $\sigma$ collapses all realizations in $\mathbb{B}$ for which $\sigma$ is infeasible, meanwhile transforming each of the other realizations
according to ruleset {\sc R}.
\end{definition}

For example, the move $(0, 0, 0, 0, -2, 0)$ applied to the quantum {\sc Nim} position in Equation 
(\ref{eqn:QuantumNimExample}) collapses realizations 2 and 4, and transforms realizations 1 and 3, according to {\sc Nim} as:

$$
\left( 
\begin{array}{ccccccc}
5 & 3 & 0 & 4 & \{2-2\} & 2\\
\boxtimes & \boxtimes & \boxtimes & \boxtimes& \boxtimes& \boxtimes \\
0 & 0 & 4 & 6 & \{5-2\} & 7\\
\boxtimes & \boxtimes & \boxtimes & \boxtimes& \boxtimes& \boxtimes \\
\end{array}
\right)
=
\left( 
\begin{array}{ccccccc}
5 & 3 & 0 & 4 & 0 & 2\\
0 & 0 & 4 & 6 & 3 & 7
\end{array}
\right)
$$

Note that for any impartial ruleset ${\sc R}$, {\sc Demi-Quantum-R} remains impartial.
We now show that {\sc Demi-Quantum-Nim} contains {\sc Crosswise AND}, and thus, {\sc Transverse Wave}, as a special case.

\begin{proposition}[QCGT Connection of {\sc Transverse Wave}]
Let {\sc Boolean Nim} denote {\sc Nim} in which each heap has either one or zero pebbles.
{\sc Demi-Quantum Boolean Nim} is isomorphic to {\sc Crosswise AND}, and
hence isomorphic to {\sc Transverse Wave}.
\end{proposition}
\begin{proof}
The proof uses the following  equivalent ``numerical interpretation'' 
of collapses in (demi-)quantum generalization of {\sc Nim}.
When a realization collapses, we can either remove it from the superposition or replace it with a {\sc Nim} position in which all piles have zero pebbles.
For example, the following two {\sc Nim} superpositions are equivalent for subsequent game dynamics. 
$$
\left( 
\begin{array}{ccccccc}
5 & 3 & 0 & 4 & 0 & 2\\
\boxtimes & \boxtimes & \boxtimes & \boxtimes& \boxtimes& \boxtimes \\
0 & 0 & 4 & 6 & 3 & 7\\
\boxtimes & \boxtimes & \boxtimes & \boxtimes& \boxtimes& \boxtimes \\
\end{array}
\right)
\equiv
\left( 
\begin{array}{ccccccc}
5 & 3 & 0 & 4 & 0 & 2\\
0 & 0  & 0 & 0& 0& 0 \\
0 & 0 & 4 & 6 & 3 & 7\\
0 & 0 & 0 & 0& 0& 0 \\
\end{array}
\right)
$$


In {\sc Boolean Nim}, 
  each move can only remove one pebble
  from a pile.
So, we can simplify the specification of the move by the index $i$ alone.
Note also that each quantum {\sc Boolean Nim} position can be specified by a Boolean matrix.
Let $\mathbf{B}$ denote the Boolean matrix of {\sc Demi-Quantum-Boolean-Nim} under consideration.
With the above numerical interpretation of collapses in 
demi-quantum generalization, the collapse of a realization of $\mathbf{B}$
when applying  a move $i$ corresponding to the case that 
the corresponding row in $\mathbf{B}$ has 0 at $i^{th}$ entry, 
and the row is replaced by the crosswise AND with that column selection.
Thus, {\sc Demi-Quantum Boolean Nim} with position $\mathbf{B}$ is 
isomorphic to  {\sc Crosswise AND} with position $\mathbf{B}$.
\end{proof}

As an aside, notice that positions with all green tiles are trivial. Interesting games need be primed with some arbitrary purpletiles. Thus  the hard positions given by the reduction could be natural starting positions. Thus the hardness statement is particularly meaningful for \ruleset{Transverse Wave}.

\section{The Graph Structures Underlying {\sc Demi-Quantum Nim}}
\label{Sec:GraphStructure}

As the basis of Nimbers and Sprague-Grundy theory \cite{WinningWays:2001,Sprague:1936,Grundy:1939}, {\sc Nim} holds a unique
place in combinatorial game theory.
It is also among the few non-trivial-looking games with polynomial-time solution.
Over the past decades, multiple efforts have been made to introduce
graph-theoretical elements into the game of {\sc Nim} \cite{DBLP:journals/tcs/Fukuyama03,StockmanREU:2004,BurkeGeorge}.
In 2001, Fukuyama \cite{DBLP:journals/tcs/Fukuyama03} introduced an edge-version of 
{\sc Graph Nim}, with {\sc Nim} piles placed on edges of undirected graphs.
Stockman \cite{StockmanREU:2004} analyzed several versions with piles on the nodes. 
Both use the graph structure to capture the locality of the piles players can nim from.
Burke and George \cite{BurkeGeorge} then formulated a version called {\sc Neighboring Nim}, for which  classical {\sc Nim} corresponds to {\sc Neighboring Nim} over the complete graph, where each vertex hosts a pile.

The graph structures have profound impact to the game of {\sc Nim}
both mathematically \cite{DBLP:journals/tcs/Fukuyama03,StockmanREU:2004} and computationally
\cite{BurkeGeorge}.
By a reduction from {\sc Geography}, Burke and George 
proved  that {\sc Neighboring Nim} on some graphs is
 PSPACE-hard while on others (such as complete graph) is polynomial-time solvable \cite{BurkeGeorge}.
 However, {\sc Neighboring Boolean Nim}, the graph {\sc Nim} where each pile has at most one pebble, is equivalent to {\sc Undirected Geography}, and thus can be solved in polynomial time \cite{DBLP:journals/tcs/FraenkelSU93}.
 

In contrast, {\sc Demi-Quantum Boolean Nim} is intractable.

\begin{theorem}[Intractability of {\sc Demi-Quantum Boolean Nim}]
{\sc Demi-Quantum Boolean Nim}, and hence {\sc Transverse Wave} \mbox{\rm (}{\sc Crosswise AND}; {\sc Crosswise OR}\mbox{\rm )}, is a PSPACE-complete game.
\end{theorem}

\subsection{The Logic and Graph Structures of {\sc Demi-Quantum Boolean Nim}}

The intractability follows from the next theorem, 
   which connects {\sc Demi-Quantum Boolean Nim} to 
   Schaefer's elegant PSPACE-complete game, {\sc Avoid True} \cite{DBLP:journals/jcss/Schaefer78}.
The reduction also reveals the bipartite and hyper-graph structures of \ruleset{Demi-Quantum Boolean Nim}.

\begin{gameruleset}[{\sc Avoid True}]
A game position of \ruleset{Avoid True} is defined by 
 a positive CNF $F$ (\texttt{and} of a set of \texttt{or}-clauses of only positive variables) over a ground set $V$ and  a subset $T\subset V$, the "true" variables, (which is usually the empty set at the beginning of the game) .
 
A turn consists of selecting one variable from $V \setminus T$, where 
a variable $x\in V \setminus T$ is {\em feasible} for position $(F,V,T)$
if assigning all variables in $T\cup\{x\}$ to \texttt{true} does not make $F$ {\texttt{true}}. 
If $x$ is feasible, then the position resulting from that move is $(F,V,T\cup\{x\})$.
Under normal play, the next player loses if the position has no feasible move.
\end{gameruleset}

\begin{theorem}[\cite{BurkeFerlandTengQCGT}]
\label{theo:AllAboutAvoidTrue}
{\sc Demi-Quantum Boolean Nim} and {\sc Avoid True} are isomorphic games.
\end{theorem}
\begin{proof} Part of the proof in \cite{BurkeFerlandTengQCGT} 
showing that {\sc Quantum Nim} is $\Sigma_2$-hard also establishes the above theorem.
Because establishing this theorem is not the main focus of \cite{BurkeFerlandTengQCGT},
we reformulate the proof here to make this theorem more explicit
as well as to provide a complete background of our discussion in this section.

We first establish the direction from {\sc Demi-Quantum Boolean Nim} to {\sc Avoid True}.
Given a position $\mathbb{B}$ in {\sc Demi-Quantum Boolean Nim}, 
we can create a  {\texttt or}-clause from each realization in $\mathbb{B}$.
Suppose $\mathbb{B}$ has $m$ realizations and $n$ piles.
We introduce $n$ Boolean variables, $V = \{x_1,...,x_n\}$.
For each realization in $\mathbb{B}$, the \texttt{or}-clause 
consists of all variables corresponding to piles with zero pebbles.
The reduced CNF $F_{\mathbb{B}}$ is the \texttt{and} of all these \texttt{or}-clauses.
Because taking a pebble from a pile collapses a realization for which the pile has no pebble is mapped to selecting the corresponding 
Boolean variable making the \texttt{or}-clause associated with the realization \texttt{true},
playing {\sc Demi-Quantum Boolean Nim} starting position 
$\mathbb{B}$ is isomorphic to playing {\sc Avoid True} starting at position $(F_{\mathbb{B}},V,\emptyset)$.
Note that the reduction can be set up in  polynomial time.

For example, consider this \ruleset{Demi-Quantum Boolean Nim} position (with heaps (columns) labelled by their indices and realizations (rows) labelled $A$, $B$, and $C$):

$$\begin{pNiceMatrix}[first-row,last-col=8] 
 1 & 2 & 3 & 4 & 5 & 6 & 7 &   \\
 1 & 0 & 0 & 1 & 1 & 0 & 1 & A \\
 0 & 1 & 0 & 1 & 1 & 1 & 0 & B \\
 1 & 0 & 0 & 1 & 1 & 1 & 0 & C \\
\end{pNiceMatrix}$$

This reduces to the \ruleset{Avoid True} position with formula:

$$\underbrace{(x_2 \vee x_3 \vee x_6 \vee x_7)}_A \wedge  \underbrace{(x_1 \vee x_3 \vee x_7)}_B \wedge \underbrace{(x_2 \vee x_3 \vee x_7)}_C $$

and $T = \emptyset$.  The three clauses are labelled by their respective realization.  Those variables that appear in each clause are those with a zero in the matrix. Notice that:

\begin{itemize}
  \item The third heap is empty in all \ruleset{Nim} realizations, so no player can legally play there.  That is the same in the resulting \ruleset{Avoid True} position; no player can pick $x_3$ as it is in all clauses and would make the formula true.
  \item Heaps 4 and 5 have a pebble in all three realizations in \ruleset{Nim}, so a player can play in either of them without any collapses.  Because of this, those Boolean variables don't occur in any of the \ruleset{Avoid True} clauses.  
\end{itemize}

For the reverse direction, 
  consider an {\sc Avoid True} position $(F,V,T)$.
Assume $V = \{x_1,...,x_n\}$, and $F$ has $m$ clauses, $C_1,...,C_m$.
We reduce it to a {\sc Boolean Nim} superposition $\mathbb{B}_{(F,V,T)}$
with $m$ realizations and $n$ piles.
In the realization for $C_i$, we set piles corresponding to variables in $C_i$ zero to set up the mapping between collapsing the realization with  making the clause \texttt{true}.
We also set all piles associated with variables in $T$ to zero, to
set up the mapping between collapsing the realization with selecting a selected variable.
Again, we can use these two mappings to inductively establish that 
  the game tree for {\sc Demi-Quantum Boolean Nim}  
  at $\mathbb{B}_{(F,V,T)}$ is isomorphic 
  to the game tree for {\sc Avoid True}
  at $(F,V,T)$.
Note that the reduction also runs in polynomial time.

We demonstrate this reduction on the following \ruleset{Avoid True} position, with formula:

$$\underbrace{(x_1 \vee x_2 \vee x_3 \vee x_4)}_A \wedge \underbrace{(x_1 \vee x_5 \vee x_6 \vee x_7)}_B \wedge \underbrace{(x_1 \vee x_3 \vee x_6)}_C \wedge \underbrace{(x_2 \vee x_5 \vee x_8)}_D$$

and already-chosen variables, $T = \{x_8\}$.  Following the reduction, we produce the following \ruleset{Demi-Quantum Boolean Nim} position:

$$\begin{pNiceMatrix}[first-row,last-col=9] 
 1 & 2 & 3 & 4 & 5 & 6 & 7 & 8 &   \\
 0 & 0 & 0 & 0 & 1 & 1 & 1 & 0 & A \\
 0 & 1 & 1 & 1 & 0 & 0 & 0 & 0 & B \\
 0 & 1 & 0 & 1 & 1 & 0 & 1 & 0 & C \\
\end{pNiceMatrix}$$

Note, since $x_8$ has already been made-true ($x_8 \in T$):

\begin{itemize}
    \item The 8th column is all zeroes, and
    \item Since $x_8$ appears in clause $D$, that clause does not have a corresponding realization in the quantum superposition (i.e. row in the matrix).
\end{itemize}

\end{proof}

Theorem \ref{theo:AllAboutAvoidTrue} presents the following 
``rechargeable'' bipartite-graph interpretation of {\sc Demi-Quantum Boolean Nim}.

\begin{gameruleset}[{\sc Rechargeable Bipartite Boolean Nim}]
The game is defined by a bipartite graph $G = (U,V,E)$  (with edges $E$ between vertex sets $U$ and $V$) and 
   a vertex $s\in U$, 
   where there is a {\sc Nim} pile of one pebble at each vertex in $U\setminus \{s\}$; piles at $s$ have $0$ pebbles.  
The game start at $s$.
During the game, the players takes turns to search a pebble reachable from
 the current vertex (initially $s$) via a node in $V$ (we call it a ``station''): 
They first cross to a node in $V$ in order to reach to a vertex in 
$s'\in U$ whose Nim pile contains a pebble. 
However, the selection of a pebble only ``recharge'' the stations  connected to $s'$.
All other stations of $V$ got powered-off and permanently lost their connections to piles in $U$.
The player who cannot find the next ``power pebble''/``energy pebble'' loses the game.
\end{gameruleset}

Theorem \ref{theo:AllAboutAvoidTrue} also gives the following simple hyper-graph interpretation of {\sc Demi-Quantum Boolean Nim}.
Recall a hyper-graph $H$ over a groundset $V=[n]$ is a collection of subsets in $V$. We write $H =(V,E)$, where for each hyper-edge $e\in E$, 
$e\subseteq V$.  

\begin{gameruleset}[{\sc  Rechargeable Hypergraph  Boolean Nim}]
The game position is defined by hyper-graph $H=(V,E)$, in which each vertex $v\in V$ has a {\sc Nim} pile of $1$ pebble, and a vertex $c\in V$.
The next player needs to move to a non-empty vertex $v$ connected to $c$ by an hyper-edge, taking its pebble which also removes all hyper-edges  without $v$ from the hyper-graph (the ``energy pebble'' only recharge hyper-edges incident to $v$).
The player who cannot find the next `pebble loses the game.
\end{gameruleset}

Because both {\sc Rechargeable Bipartite Boolean Nim} and 
{\sc  Rechargeable Hypergraph  Boolean Nim} of simply the graph-theoretical interpretation of the proof for Theorem \ref{theo:AllAboutAvoidTrue}, the proof also establish the following:

\begin{corollary}
{\sc Rechargeable Bipartite Boolean Nim} and 
{\sc  Rechargeable Hypergraph  Boolean Nim} are isomorphic to 
{\sc Transverse Wave} and are therefore PSPACE-complete combinatorial games.
\end{corollary}

\subsection{A Social Influence Game Motivated by {\sc Demi-Quantum Nim}}

\label{Sec:DI}

The connection between {\sc Demi-Quantum Boolean Nim} and {\sc Friend Circle} motivates the following social-influence-inspired game, 
{\sc Demographic Influence}, which mathematically generalizes 
{\sc Demi-Quantum Nim}.
In a nutshell, the setting has a demographic structure 
  over a population, in which individuals have their own friend circles.
Members in the population are initially \textit{un-influenced} but \textit{receptive} to ads, and (viral marketing) 
influencers try to target their ads at demographic groups to influence
the population.
People can be influenced either by influencers' ads or by 
``enthusiastic endorsement'' cascading through friend circles.

The following combinatorial game is distilled from the above scenario.

\begin{gameruleset}[{\sc Demographic Influence}], and
    \item $D = \{D_1,...,D_m\}$ is a collection of subsets of $V$, representing the demographic groups, which can overlap. if $\Theta(v) < 0$, then $v$ is \emph{strongly influenced}.
A {\sc Demographic Influence} position
  is defined by a tuple $Z = (G, \Theta, D)$, where
\begin{itemize}
  \item $G = (V, E)$ is an undirected graph representing a symmetric social network on a population.
  \item $\Theta: V\rightarrow \mathbb{Z}$ represents how resistant each individual is to the product.  (I.e., their threshold to being influenced.)
  \begin{itemize}
    \item If $\Theta(v) > 0$, then $v$ is \emph{uninfluenced},
    \item if $\Theta(v) = 0$, then $v$ is \emph{weakly influenced}, and
    \item if $\Theta(v) < 0$, then $v$ is \emph{strongly influenced}.
   \end{itemize}
    \item $D = \{D_1, \ldots, D_m\}$ is the set of demographics, each a subset of $V$.
\end{itemize}

A player's turn consists of choosing a demographic, $D_k$ and the amount they want to influence, $c > 0$ where $\exists v \in D_k$ where $\Theta(v) \geq c$.  (Since $c > 0$, there must be an uninfluenced member of $D_k$.)

\begin{itemize}
  \item $\Theta(v)$ decreases by $c$ (all individuals are influenced by $c$) and
  \item If $\Theta(v)$ became negative by this subtraction (if it went from $\mathbb{Z}^+ \cup \{0\}$ to $\mathbb{Z}^-$), then $\forall x \in N_G(v): \Theta(x) = -1$.  (Thematically, if one individual becomes strongly influenced directly by the marketing campaign, then they enthusiastically recommend it to their friends and strongly influence them as well).
\end{itemize}
Importantly, we perform all the subtractions and determine which individuals are newly-strongly influenced, before they go and strongly influence their friends.

Note that when influencing a demographic, $D_k$, since $c$ cannot be greater than the highest threshold, that highest-threshold individual will not be strongly influenced by the subtraction step.  (If one of their neighbors does get strongly influenced, then they will be strongly influenced in that manner.)

Since a player needs to make a move on a demographic group with at least one uninfluenced individual, the game ends when there are no remaining groups to influence.

\end{gameruleset}

\begin{figure}[h!]
\begin{center}\begin{tikzpicture}
  \node[draw, circle] (v1) [label=above:{$v_1$}] {7};
  \node[draw, circle] (x0b) [above right=of v1] {0};
  \node[draw, circle] (xm1) [above left=of v1] {$-1$};
  \node[draw, circle] (xm2) [left=of v1] {-2};
  \node[draw, circle] (x4) [below=of v1] {4};
  \node[draw, circle] (x2) [below right=of v1] {2};
  \node[draw, circle] (v3) [below=of x2, label=below:{$v_3$}] {4};
  \node[draw, circle] (x3b) [left=of v3] {3};
  \node[draw, circle] (v2) [left=of x4, label=right:{$v_2$}] {2};
  \node[draw, circle] (x6) [below=of v2] {6};
  
  \node[] (lbl) [below=of x3b] {a};
  
  \draw[-]
    (v1) edge (x0b)
    (v1) edge (xm1)
    (v1) edge (xm2)
    (v1) edge (x4)
    (v1) edge (x2)
    (v3) edge (x2)
    (v3) edge (x4)
    (v3) edge (x3b)
    (v2) edge (x3b)
    (v2) edge (x6)
    (v2) edge (xm2)
    ;
\end{tikzpicture}\hspace{.5cm}\begin{tikzpicture}
  \node[draw, circle] (v1) [label=above:{$v_1$}] {3};
  \node[draw, circle] (x0b) [above right=of v1] {0};
  \node[draw, circle] (xm1) [above left=of v1] {$-1$};
  \node[draw, circle] (xm2) [left=of v1] {-2};
  \node[draw, circle] (x4) [below=of v1] {4};
  \node[draw, circle] (x2) [below right=of v1] {2};
  \node[draw, circle] (v3) [below=of x2, label=below:{$v_3$}] {0};
  \node[draw, circle] (x3b) [left=of v3] {3};
  \node[draw, circle] (v2) [left=of x4, label=right:{$v_2$}] {-2};
  \node[draw, circle] (x6) [below=of v2] {6};
  
  \node[] (lbl) [below=of x3b] {b};
  
  \draw[-]
    (v1) edge (x0b)
    (v1) edge (xm1)
    (v1) edge (xm2)
    (v1) edge (x4)
    (v1) edge (x2)
    (v3) edge (x2)
    (v3) edge (x4)
    (v3) edge (x3b)
    (v2) edge (x3b)
    (v2) edge (x6)
    (v2) edge (xm2)
    ;
\end{tikzpicture}\hspace{.5cm}\begin{tikzpicture}
  \node[draw, circle] (v1) [label=above:{$v_1$}] {3};
  \node[draw, circle] (x0b) [above right=of v1] {0};
  \node[draw, circle] (xm1) [above left=of v1] {$-1$};
  \node[draw, circle] (xm2) [left=of v1] {-1};
  \node[draw, circle] (x4) [below=of v1] {4};
  \node[draw, circle] (x2) [below right=of v1] {2};
  \node[draw, circle] (v3) [below=of x2, label=below:{$v_3$}] {0};
  \node[draw, circle] (x3b) [left=of v3] {-1};
  \node[draw, circle] (v2) [left=of x4, label=right:{$v_2$}] {-2};
  \node[draw, circle] (x6) [below=of v2] {-1};
  
  \node[] (lbl) [below=of x3b] {c};
  
  \draw[-]
    (v1) edge (x0b)
    (v1) edge (xm1)
    (v1) edge (xm2)
    (v1) edge (x4)
    (v1) edge (x2)
    (v3) edge (x2)
    (v3) edge (x4)
    (v3) edge (x3b)
    (v2) edge (x3b)
    (v2) edge (x6)
    (v2) edge (xm2)
    ;
\end{tikzpicture}\end{center}

\caption{A \ruleset{Demographic Influence} move, influencing $D_3 = \{v_1, v_2, v_3\}$ by 4.  (a) shows $G$, $\Theta$, and $D_3$ prior to making the move.  (b) shows the first part of the move: subtracting from the thresholds of $v_1$, $v_2$, and $v_3$.  (c) since $v_2$ went negative, it's neighbors are set to $-1$ to show that they have been strongly influenced as well.  (The magnitude of negativity doesn't matter, so it's okay that the vertex at -2 goes back to -1.)}
\end{figure}
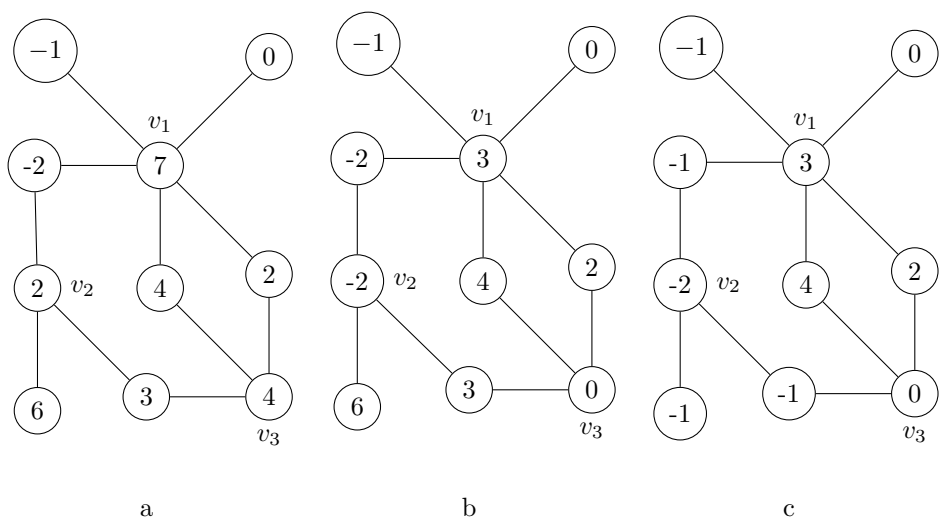

The following theorem shows that {\sc Demographic Influence} 
generalizes {\sc Demi-Quantum Nim}.

\begin{theorem}[{\sc Demi-Quantum Nim} Generalization: Social-Influence~Connection] \label{Theo:DemographicInfluence}
{\sc Demographic Influence} contains {\sc Demi-Quantum Nim} as a Special Case.
Therefore, {\sc Demographic Influence} is a \cclass{PSPACE}-complete game.
\end{theorem}
\begin{proof}
For every {\sc Demi-Quantum Nim} instance $Z$ with $m$ realizations of $n$ piles, we construct the following
{\sc Demographic Influence} instance $Z'$, in which,
(1) $V = \{(r,c)\ | \ r\in [m], c\in [n]\}$.
(2) $E = \{((r_1, c_1), (r_2, c_2))\ |\ r_1 = r_2\}$ (i.e. vertices from all piles from the same realization are a clique),
(3) $\forall (r,c)\in V$, $\Theta((r,c))$ is set to be the number of pebbles that $c^{th}$ pile has in $r^{th}$ realization of {\sc Nim}.
(4) $D = (D_1,...,D_n)$, where $D_c = \{(r,c)\ | \ r\in [m]\}$, i.e., nodes associated with the $c^{th}$ {\sc Nim} pile.

We claim that $Z$ and $Z'$ are isomorphic games.
Imagine the two games are played in tandem.
Suppose the player in {\sc Demi-Quantum Nim} $Z$ makes a move $(k,q)$, that is, removing $q$ pebbles from pile $k$.
In its {\sc Demographic Influence} counterpart, $Z'$, the corresponding player also plays $(k,q)$, meaning to invest $q$ units in demographic group
$k$.
Note that in $Z$, $(k,q)$ is feasible iff in at least one of the realizations, the $k^{th}$ {\sc Nim} pile has at least $q$ pebbles.
This is same as $q \leq \max_{i\in [m]} \Theta((i,q))$.
Therefore, $(k,q)$ is feasible in $Z$ iff $(k,q)$ is feasible in $Z'$.

When $(k,q)$ is feasible, then for any realization $i\in [m]$, 
there are three cases:
(1) if the $k^{th}$ {\sc Nim} pile has more pebbles than $q$, then 
in that realization, a classical transition is made, i.e., the $q$ pebbles are removed from the pile.
This corresponds to the reduction of the threshold at node $(i,k)$ by $q$. 
(2) if the $k^{th}$ {\sc Nim} pile has exactly $q$ pebbles,
then all pebbles are removed from the pile.
This corresponds to the case that node $(i,k)$ becomes weakly influenced.
(3) if $q$ is more than the number of pebbles in the $k^{th}$ {\sc Nim} pile, then the move collapses realization $i$.
This corresponds to the case in {\sc Demographic Influence}
where $(i,k)$ become strongly influenced, and then strongly influences
all other vertices in the same row ($i$).
Therefore, $Z$ and $Z'$ are isomorphic games, with connection 
between the collapse of a realization in the quantum version 
and the cascading of influence by endorsement in friend circle. 
\end{proof}


The proof of Theorem \ref{Theo:DemographicInfluence} illustrates that
 {\sc Demographic Influence} can be viewed as a graph-theoretical extension of {\sc Nim}.
Recall Burke-George's {\sc Neighboring Nim}, which extends both 
  the classical {\sc Nim} (when the underlying graph is a clique)
  and {\sc Undirected Geography} (when all {\sc Nim} heaps have at most one item in them, i.e. \ruleset{Boolean Nim}).
The next theorem complements Theorem \ref{Theo:DemographicInfluence}
by showing that {\sc Demographic Influence} also 
generalizes {\sc Node-Kayles}.

\begin{theorem}[Social-Influence Connection with {\sc Node-Kayles}]
{\sc Demographic Influence} contains  {\sc Node-Kayles} as a special case.
\end{theorem}
\begin{proof}
Consider a {\sc Node-Kayles} instance defined by an $n$-node undirected graph $G_0 = (V_0,E_0)$ with $V_0 = [n]$.
We define a {\sc Demographic Influence} instance 
$Z=(G,\Theta,D)$ as the following.
(1) For each $v\in V_0$, we introduce a new vertex $t_v$. 
Let $V= V_0\cup T_0$, where $T_0 :=\{t_v\ |\ v\in V_0\}$.
(2) For all $v\in V$, $\Theta(v) = 0$ and $\Theta(t_v) = 1$.
(3) For all $v\in V$, 
$C(v) = N_G(v) \cup \{t_w\ |\ w \in N_G(v)\}\cup \{t_v\}$ and $C(t_v)=\{v\}$.
(4)$D = \{D_1,...,D_n\}$, where $D_v = \{v, t_v\}, \forall v\in V$.

Note that because  $\Theta(v)\in \{0,1\}, \forall v\in V$, the
space of moves in this {\sc Demographic Influence}
is $\{(v,1)\ | v\in [n]\}$, whereas in {\sc Node-Kayles}, a move consists of selecting one of the vertices from $[n]$.

We now show that {\sc Demographic Influence} over $Z$ is isomorphic to
{\sc Node-Kayles} over $G$ under the mapping of moves $(v, 1) \Leftrightarrow v$.

For a \ruleset{Node Kayles} move at $v$, it removes $v$ and all $x \in N_{G_0}(v)$ from future move choices.  In \ruleset{Demographic Influence}, choosing the $(v, 1)$ move means that $\Theta(t_v)$ becomes 0 and $\Theta(v)$ becomes -1.  $v$ is now strongly influenced and they strongly influence their neighbors at $N_G(v) = N_{G_0}(v) \cup \{t_x\ |\ x \in N_{G_0}(v)\}$, so all those vertices also gain a threshold of -1.  Those include both the $x$ and $t_x$ vertices for each $x \in N_{G_0}(v)$, so it removes all those neighboring demographics from future moves, $(x, 1)$, the set of which is isomorphic to those removed from the corresponding \ruleset{Node Kayles} move.

Therefore, with induction, we establish that
{\sc Demographic Influence} over $Z$ is isomorphic to
{\sc Node-Kayles} over $G$ under the mapping
 of moves from $v \Leftrightarrow (v,1)$.
\end{proof}

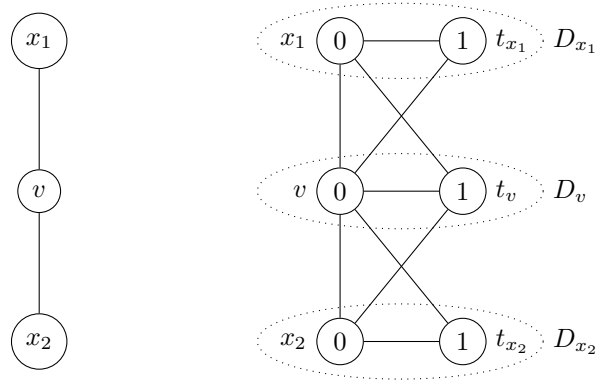
\begin{figure}[h!]
\begin{center}\begin{tikzpicture}
  \node[draw, circle] (v)  {$v$};
  \node[draw, circle] (x1) at (0, 2) {$x_1$};
  \node[draw, circle] (x2) at (0, -2) {$x_2$};
  
  \draw
    (v) edge (x1)
    (v) edge (x2)
    ;
  
  \node[draw, circle] (vb) at (4, 0) [label=left:{$v$}] {0};
  \node[draw, circle] (tv) [right=of vb, label=right:{$t_v$}] {1};
  \node[draw, ellipse, minimum height=1cm, minimum width=3.8cm, dotted] at (4.8, 0) [label=right:{$D_v$}] {};
  \node[draw, circle] (x1b) at (4, 2) [label=left:{$x_1$}] {0};
  \node[draw, circle] (tx1) [right=of x1b, label=right:{$t_{x_1}$}] {1};
  \node[draw, ellipse, minimum height=1cm, minimum width=3.8cm, dotted] at (4.8, 2) [label=right:{$D_{x_1}$}] {};
  \node[draw, circle] (x2b) at (4, -2) [label=left:{$x_2$}] {0};
  \node[draw, circle] (tx2) [right=of x2b, label=right:{$t_{x_2}$}] {1};
  \node[draw, ellipse, minimum height=1cm, minimum width=3.8cm, dotted] at (4.8, -2) [label=right:{$D_{x_2}$}] {};
  
  \draw
    (vb) edge (tv)
    (vb) edge (tx1)
    (vb) edge (tx2)
    (vb) edge (x1b)
    (vb) edge (x2b)
    (x1b) edge (tx1)
    (x2b) edge (tx2)
    (tv) edge (x1b)
    (tv) edge (x2b)
    ;
\end{tikzpicture}\end{center}

\caption{Example of reduction from \ruleset{Node Kayles} to \ruleset{Demographic Influence}.  }
\label{fig:kaylesDemographicInfluence}
\end{figure}

Therefore, {\sc Demographic Influence} simultaneously generalizes
{\sc Node-Kayles} and {\sc Demi-Quantum Nim} (which in turn generalizes the classical {\sc Nim}, {\sc Avoid True}, and {\sc Transverse Wave}).

We can also establish that {\sc Demographic Influence} Generalizes {\sc Friend Circle} defined in the earlier section.

\begin{theorem}[{\sc Demographic Influence} Generalizes {\sc Friend Circle}]
{\sc Demographic Influence} contains  {\sc Friend Circle} as a special case.
\end{theorem}
\begin{proof}
For a {\sc Friend Circle} position, $Z$ = $(G,S,w)$, where $G = (V,E)$,
we construct the following {\sc Demographic Influence} 
instance  $Z'= (G', \Theta, D)$ as the following:
\begin{itemize}
  \item For each edge $e\in E$, we create a new vertex $v_e$. Then $V' = \{v_e\ |\ e \in E\}$
  \item $E' = \{(v_{e_1}, v_{e_2})\ |\ \exists v \in V: e_1, e_2 \text{ both incident to } v \}$
  \item $G' = (V', E')$
  \item $\Theta: V' \rightarrow \{0, 1\}$, where if $w(e) = \fcFalse$, then $\Theta(e)=1$; otherwise, if  $w(e) = \fcTrue$, 
then $\Theta(e) = 0$. 
  \item $D = \{D_s\}_{s \in S}$, where $D_s = \{ v_e\ |\ e \text{ is incident to } s \}$

\end{itemize}
In other words, the connections
are built on the line graph 
of the underlying graph in {\sc Friend Circle}.
Each seed vertex defines the demographic group and associates with all
edges incident to it.
Targeting this demographic group
  influences all these edges and 
  edges adjacent to $\fcTrue$-edges 
  in this set.
  
We complete the proof by showing that a play on \ruleset{Friend Circle} position $s$ is isomorphic to a play on \ruleset{Demographic Influence} $(D_s, 1)$, meaning choosing demographic $D_s$ and investing $c = 1$.  In \ruleset{Friend Circle}, playing at $s$ means that $\forall e$ incident to $s:$ (1) $w(e)$ becomes \fcTrue, and (2) if $w(e)$ was already \fcTrue, then $\forall f$ adjacent to $e: w(f)$ becomes \fcTrue.

In \ruleset{Demographic Influence}, the corresponding play, $(D_s, 1)$ means that $\forall v_e \in D_s:$ \begin{itemize}
  \item $\Theta(v_e)$ is reduced by 1.  This corresponds to setting $w(e)$ to \fcTrue.
  \item If $\Theta(v_e)$ becomes -1, then $\forall v_f \in N_{G'}(v_e): \Theta(v_f)$ also becomes $-1$.  By our definition of $E'$, these $v_f$ are exactly those where both $w(e)$ was previously \fcTrue and $e$ and $f$ are adjacent in $G$.
\end{itemize}
Thus, following analagous moves, $w(e) =$  \fcTrue\ iff $\Theta(v_e) \leq 0$.  A seed vertex, $s' \in S$ is surrounded by \fcTrue edges (and inelligible as a move) exactly when all vertices $v_e \in D_s$ are influenced, also making $D_s$ inelligible as a move.  Thus, this mapping of moves shows that the games are isomorphic.

\end{proof}

See Figure \ref{fig:friendCircleToDemographicInfluence} for an example of the reduction.

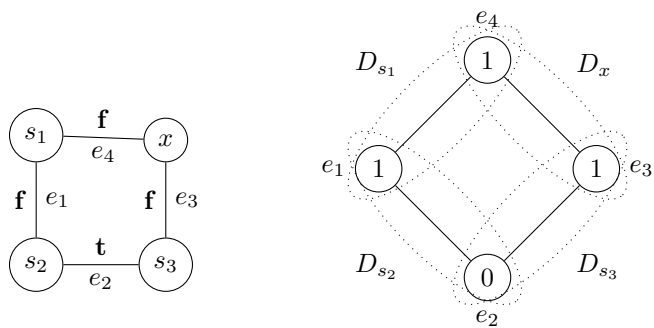
\begin{figure}[h!]
\begin{center}\begin{tikzpicture}
  \node[draw, circle] (s1) {$s_1$};
  \node[draw, circle] (s2) [below=of s1] {$s_2$};
  \node[draw, circle] (s3) [right=of s2] {$s_3$};
  \node[draw, circle] (x) [above=of s3] {$x$};
  
  \draw[-]
    (s1) edge node [left] {\fcFalse} node [right] {$e_1$} (s2)
    (s2) edge node [above] {\fcTrue} node [below] {$e_2$} (s3)
    (s3) edge node [left] {\fcFalse} node [right] {$e_3$} (x)
    (x) edge node [above] {\fcFalse} node [below] {$e_4$} (s1)
    ;
    
  \node[draw, circle] (e4) at (6, 1) [label=above:{$e_4$}] {1};
  \node[draw, circle] (e1) [below left=of e4, label=left:{$e_1$}] {1};
  \node[draw, circle] (e2) [below right=of e1, label=below:{$e_2$}] {0};
  \node[draw, circle] (e3) [above right=of e2, label=right:{$e_3$}] {1};
  
  \draw[-]
    (e1) edge (e2)
    (e2) edge (e3)
    (e3) edge (e4)
    (e4) edge (e1)
  ;
  
  \node[draw, ellipse, minimum height=1cm, minimum width=3.1cm, dotted, rotate=-45] at (5.3, -1.1) [label=below:{$D_{s_2}$}] {};
  \node[draw, ellipse, minimum height=1cm, minimum width=3.1cm, dotted, rotate=45] at (5.3, .3) [label=above:{$D_{s_1}$}] {};
  \node[draw, ellipse, minimum height=1cm, minimum width=3.1cm, dotted, rotate=-45] at (6.7, .3) [label=above:{$D_x$}] {};
  \node[draw, ellipse, minimum height=1cm, minimum width=3.1cm, dotted, rotate=45] at (6.7, -1.1) [label=below:{$D_{s_3}$}] {};
    
\end{tikzpicture}\end{center}

\caption{Example of the reduction.  On the left is a \ruleset{Friend Circle} instance. On the right is the resulting \ruleset{Demographic Influence} position.}
\label{fig:friendCircleToDemographicInfluence}
\end{figure}

\section{Conclusions and Future Work}

One of the beautiful aspects of Winning Ways \cite{WinningWays:2001} is the relationships between games, especially when positions in one ruleset can be transformed into equivalent instances of another ruleset.  As examples, \ruleset{Dawson's Chess} positions are equivalent to \ruleset{Node Kayles} positions on paths, \ruleset{Wythoff's Nim} is the one-queen case of \ruleset{Wyt Queens}, and \ruleset{Subtraction}-$\{1, 2, 3, 4\}$ positions exist as instances of many rulesets, including \ruleset{Adders and Ladders} with one token after the top of the last ladder and last snake.

Transforming instances of one ruleset to another (reductions) is a basic part of Combinatorial Game Theory\footnote{"Change the Game!" is the title of a section in Lessons in Play, Chapter 1, "Basic Techniques" \cite{LessonsInPlay:2007}.}, just as it is vital to computational complexity.  \ruleset{Transverse Wave} arose not only by reducing to other things, but more concretely as a special case of other games we explored. 

"Ruleset $A$ is a special case of ruleset $B$" (i.e. "$B$ is a generalization of $A$") not only proves that computational hardness of $A$ results in the computational hardness of $B$, but also:

\begin{itemize}
  \item If $A$ is interesting, then it is a fundamental part of $B$'s strategies, and
  \item If $B$'s rules are straightforward, then $A$ could be a block to create other fun rulesets.
\end{itemize}

Relevant special-case/generalization relationships between games presented here include:
\begin{itemize}
  \item \ruleset{Demi-Quantum Nim} is a generalization of \ruleset{Nim}. (See Section \ref{Sec:Quantum}.)  (This is true of any ruleset $R$ and \ruleset{Demi-Quantum $R$}.)
  \item \ruleset{Demi-Quantum Nim} is a generalization of \ruleset{Transverse Wave}. (Via \ruleset{Demi-Quantum Boolean Nim}, see section \ref{Sec:DQNTW}.)
  \item \ruleset{Friend Circle} is a generalization of both \ruleset{Transverse Wave} (Section \ref{Sec:TWFC}) and 
  \ruleset{Node Kayles} (Section \ref{Sec:NKFC}).
  \item \ruleset{Demographic Influence} is a generalization of both \ruleset{Demi-Quantum-Nim}  and \ruleset{Friend Circle} (Section \ref{Sec:DI}).
\end{itemize}

We show these relationships in a lattice-manner in Figure \ref{fig:generalizationLattice}.  Understanding \ruleset{Transverse Wave} is a key piece of the two other rulesets, which also include the impartial classics \ruleset{Nim} and \ruleset{Node Kayles}.

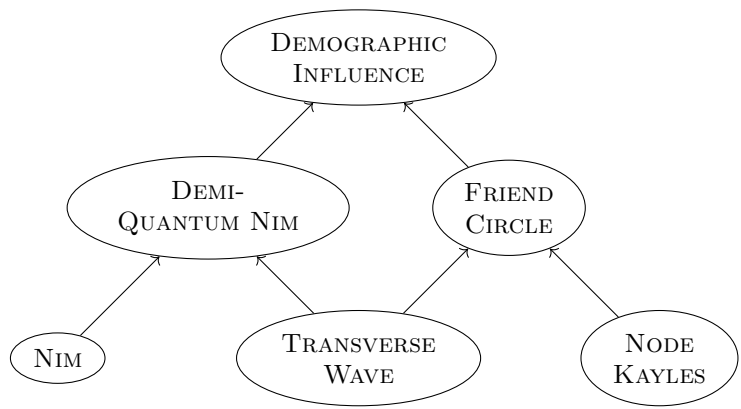
\begin{figure}[h!]
\begin{center}\begin{tikzpicture}[node distance = 1cm]

  \node[draw, ellipse] (nim) {\ruleset{Nim}};
  \node[draw, ellipse, align=center] (twave) at (4, 0) {\ruleset{Transverse}\\ \ruleset{Wave}};
  \node[draw, ellipse, align=center] (nodeKayles) at (8,0) {\ruleset{Node}\\ \ruleset{Kayles}};
  
  \node[draw, ellipse, align=center] (dqNim) at (2,2) {\ruleset{Demi-}\\ \ruleset{Quantum Nim}}; 
  \node[draw, ellipse, align=center] (friendCircle) at (6,2) {\ruleset{Friend}\\ \ruleset{Circle}};
  
  \node[draw, ellipse, align=center] (demographicInfluence) at (4, 4) {\ruleset{Demographic}\\ \ruleset{Influence}};
  
  \draw[->]
      (nim) edge (dqNim)
      (twave) edge (dqNim)
      (twave) edge (friendCircle)
      (nodeKayles) edge (friendCircle)
      (dqNim) edge (demographicInfluence)
      (friendCircle) edge (demographicInfluence)
      ;
\end{tikzpicture}\end{center}
\caption{Generalization relationships of the rulesets in this paper.  $A \rightarrow B$ means that $A$ is a special case of $B$ and $B$ is a generalization of $A$.}
\label{fig:generalizationLattice}
\end{figure}

Furthermore, several of the relationships outside Figure \ref{fig:generalizationLattice} that were discussed in this paper were completely isomorphic, preserving the 
the game values not just winnability (as in Section \ref{CGT Values}). More explicitly, any new findings on the 
Grundy values for \ruleset{Transverse Wave} also give those exact same results for \ruleset{Crosswise OR}, \ruleset{Crosswise AND}, \ruleset{Demi-Quantum Boolean Nim}, \ruleset{Avoid True}, \ruleset{Rechargable Bipartite Boolean Nim}, and \ruleset{Rechargeable Hypergraph Boolean Nim}.

We have been drawn to {\sc Transverse Wave} not only because it is colorful, approachable, and intriguing, but also because its relationships with other games have inspired us to discover more connections among games. 
Our work offers us a glimpse of 
 the \textit{lattice order} induced by \textit{special-case/generalization relationships} over mathematical games, 
 which we believe 
 is an instrumental framework for both the design and comparative analysis of combinatorial games. 
In one direction of this lattice, 
  when given two combinatorial games $A$ and $B$, 
  it is a stimulating and creative process to design a game with the simplest ruleset that generalizes both $A$ and $B$.\footnote{It is also a relevant pedagogical question to ask when introducing students to combinatorial game theory.}
For example, in generalizing both {\sc Nim} and {\sc Undirected Geography}, {\sc Neighboring Nim} highlights the role of 
``self-loop'' in {\sc Graph-Nim}.
In our work, the aim to capture both {\sc Node Kalyes} and {\sc Demi-Quantum Nim} has contributed to our design of {\sc Demographic Influence}.
In the other direction, identifying a well-formulated basic game at the intersection of two seemingly unconnected games may greatly expand our understanding of game structures. 
It is also a refinement process for identifying intrinsic building blocks and fundamental games.
By exploring the lattice order of game relationships,
we will continue to 
improve our understanding of combinatorial game theory and identify new fundamental games inspired by the rapidly evolving world of data, network, and computing.

\bibliographystyle{plain} 
\bibliography{paithan} 

\begin{thebibliography}{10}

\bibitem{LessonsInPlay:2007}
M.~H. Albert, R.~J. Nowakowski, and D.~Wolfe.
\newblock {\em Lessons in Play: An Introduction to Combinatorial Game Theory}.
\newblock A. K. Peters, Wellesley, Massachusetts, 2007.

\bibitem{WinningWays:2001}
Elwyn~R. Berlekamp, John~H. Conway, and Richard~K. Guy.
\newblock {\em Winning Ways for your Mathematical Plays}, volume~1.
\newblock A K Peters, Wellesley, Massachsetts, 2001.

\bibitem{Bouton:1901}
Charles~L. Bouton.
\newblock Nim, a game with a complete mathematical theory.
\newblock {\em Annals of Mathematics}, 3(1/4):pp. 35--39, 1901.

\bibitem{BurkeFerlandTengQCGT}
Kyle Burke, Matthew Ferland, and Shang{-}Hua Teng.
\newblock Quantum combinatorial games: Structures and computational complexity.
\newblock {\em CoRR}, abs/2011.03704, 2020.

\bibitem{BurkeGeorge}
Kyle~W. Burke and Olivia George.
\newblock A pspace-complete graph nim.
\newblock {\em CoRR}, abs/1101.1507, 2011.

\bibitem{CTZGraphBasisSocialInfluence}
Wei Chen, Shang{-}Hua Teng, and Hanrui Zhang.
\newblock A graph-theoretical basis of stochastic-cascading network influence:
  Characterizations of influence-based centrality.
\newblock {\em Theor. Comput. Sci.}, 824-825:92--111, 2020.

\bibitem{dorbec2017toward}
Paul Dorbec and Mehdi Mhalla.
\newblock Toward quantum combinatorial games.
\newblock {\em arXiv preprint arXiv:1701.02193}, 2017.

\bibitem{Eppstein}
D.~Eppstein.
\newblock Computational complexity of games and puzzles, 2006.
\newblock http://www.ics.uci.edu/$\sim$eppstein/cgt/hard.html.

\bibitem{EvenTarjanHex}
S.~Even and R.~E. Tarjan.
\newblock A combinatorial problem which is complete in polynomial space.
\newblock {\em J. ACM}, 23(4):710–719, October 1976.

\bibitem{DBLP:journals/tcs/FraenkelSU93}
Aviezri~S. Fraenkel, Edward~R. Scheinerman, and Daniel Ullman.
\newblock Undirected edge geography.
\newblock {\em Theor. Comput. Sci.}, 112(2):371--381, 1993.

\bibitem{DBLP:journals/tcs/Fukuyama03}
Masahiko Fukuyama.
\newblock A nim game played on graphs.
\newblock {\em Theor. Comput. Sci.}, 1-3(304):387--399, 2003.

\bibitem{Gale:1974}
David Gale.
\newblock A curious nim-type game.
\newblock {\em American Mathematical Monthly}, 81:876--879, 1974.

\bibitem{Gale:1979}
David Gale.
\newblock The game of {H}ex and the {B}rouwer fixed-point theorem.
\newblock {\em American Mathematical Monthly}, 10:818--827, 1979.

\bibitem{goff2006quantum}
Allan Goff.
\newblock Quantum tic-tac-toe: A teaching metaphor for superposition in quantum
  mechanics.
\newblock {\em American Journal of Physics}, 74(11):962--973, 2006.

\bibitem{Grundy:1939}
P.~M. Grundy.
\newblock Mathematics and games.
\newblock {\em Eureka}, 2:198---211, 1939.

\bibitem{KKT}
David Kempe, Jon Kleinberg, and Eva Tardos.
\newblock {Maximizing the spread of influence through a social network}.
\newblock In {\em {KDD '03}}, pages 137--146. ACM, 2003.

\bibitem{LichtensteinSipser:1980}
David Lichtenstein and Michael Sipser.
\newblock Go is polynomial-space hard.
\newblock {\em J. ACM}, 27(2):393--401, 1980.

\bibitem{NashHex}
John~F. Nash.
\newblock {\em Some Games and Machines for Playing Them}.
\newblock RAND Corporation, Santa Monica, CA, 1952.

\bibitem{Reisch:1981}
S.~Reisch.
\newblock Hex ist {PSPACE}-vollst{\"a}ndig.
\newblock {\em Acta Inf.}, 15:167--191, 1981.

\bibitem{RichardsonDomingos}
Matthew Richardson and Pedro Domingos.
\newblock Mining knowledge-sharing sites for viral marketing.
\newblock In {\em Proceedings of the 8th ACM SIGKDD International Conference on
  Knowledge Discovery and Data Mining}, KDD '02, pages 61--70, 2002.

\bibitem{DBLP:journals/jcss/Schaefer78}
Thomas~J. Schaefer.
\newblock On the complexity of some two-person perfect-information games.
\newblock {\em Journal of Computer and System Sciences}, 16(2):185--225, 1978.

\bibitem{Sprague:1936}
R.~P. Sprague.
\newblock \"{U}ber mathematische {K}ampfspiele.
\newblock {\em T\^{o}hoku Mathematical Journal}, 41:438---444, 1935-36.

\bibitem{StockmanREU:2004}
Gwendolyn Stockman.
\newblock Presentation: The game of nim on graphs: Nim{G}, 2004.
\newblock Available at
  \url{http://www.aladdin.cs.cmu.edu/reu/mini_probes/papers/final_stockman.ppt}.

\end{thebibliography}


\end{document}